\documentclass[twocolumn,10pt]{IEEEtran}
\usepackage{braket}
\usepackage{dsfont}
\input{def.tex}
\usepackage{dsfont}
\usepackage{minibox,physics}
\DeclareSymbolFont{matha}{OML}{txmi}{m}{it}% txfonts
\DeclareMathSymbol{\varv}{\mathord}{matha}{118}
\usepackage{eurosym}
\usepackage{multicol}
\usepackage{eurosym}

\makeatletter

\usepackage{tikz}
\usepackage{array}
\usepackage{booktabs,adjustbox}

\newcommand{\finalcells}[2]{%
	\begingroup\sbox0{\begin{minipage}{3cm}\raggedright#1\end{minipage}}%
	\sbox2{\begin{minipage}{3cm}\raggedright#2\end{minipage}}%
	\xdef\finalheight{\the\dimexpr\ht0+\dp0+\smallskipamount\relax}%
	\xdef\finalheightB{\the\dimexpr\ht2+\dp2+\smallskipamount\relax}%
	\ifdim\finalheightB>\finalheight
	\global\let\finalheight\finalheightB
	\fi\endgroup
	\begin{minipage}[t][\finalheight][t]{3cm}\raggedright#1\end{minipage}&
	\begin{minipage}[t][\finalheight][t]{3cm}\raggedright#2\end{minipage}}

\IEEEoverridecommandlockouts
\begin{document}
	\title{Scalable Cell-Free Massive MIMO Systems: Impact of Hardware Impairments}
	\author{Anastasios Papazafeiropoulos, Emil Bj{\"o}rnson, Pandelis Kourtessis, Symeon Chatzinotas, John M. Senior \vspace{-5mm} \\
		% \author{Anastasios Papazafeiropoulos, Pandelis Kourtessis, and John M. Senior \vspace{2mm} \\
		% $^{*}$Institute for Digital Communications, University of Edinburgh, Edinburgh EH9 3JL, U.K.\\
		% $^{\dag}$SnT - securityandtrust.lu, University of Luxembourg, Luxembourg\\
		% Email: \{a.papazafeiropoulos, t.ratnarajah\}@ed.ac.uk, \{shree.sharma, symeon.chatzinotas\}@uni.lu
		% \thanks{A. Papazafeiropoulos, P. Kourtessis, and J.M. Senior are with Optical Networks Research Group,
		% University of Hertfordshire, Hatfield, AL 10 9AB, U. K. (email: \{a.papazafeiropoulos,p.kourtessis,j.m.senior\}@herts.ac.uk). }}
		% \thanks{This work was supported by the U.K. Engineering and Physical Sciences Research Council (EPSRC) under grant EP/L025299/1.}
		%}\maketitle
		% 
		% %%%%%%%%%%%%%%%%%%%%%%%%%%%%%%%%%%%%%%%%%%%%%%%%%%%%%%%%%%%%%%%%%%%%%
		% \vspace{-1cm}
		\thanks{A. Papazafeiropoulos is with the Communications and Intelligent Systems Research Group, University of Hertfordshire, Hatfield AL10 9AB, U. K., and with the SnT at the University of Luxembourg, Luxembourg. E. Bj{\"o}rnson is with the KTH Royal Institute of Technology, 16440 Kista, Sweden, and Link{\"o}ping University, 58183 Link{\"o}ping, Sweden. P. Kourtessis and John M.Senior are with the Communications and Intelligent Systems Research Group, University of Hertfordshire, Hatfield AL10 9AB, U. K., S. Chatzinotas is with the SnT at the University of Luxembourg, Luxembourg. E-mails: tapapazaf@gmail.com, emilbjo@kth.se, \{p.kourtessis,j.m.senior\}@herts.ac.uk, symeon.chatzinotas@uni.lu.}
		\thanks{Parts of this work were presented at the IEEE International Symposium on Personal, Indoor and Mobile Radio Communications (PIMRC), London, U.K., Sep. 2020~\cite{Papazafeiropoulos2020}.}}
	% } \thanks{This work has received funding from the European Research Council (ERC) under the European Union's Horizon H2020 research and innovation programme (Grant agreement No 742648).}}
	\maketitle

	\begin{abstract}	
		Standard cell-free (CF) multiple-input-multiple-output (mMIMO) systems is a promising technology to cover the demands for higher data rates in fifth-generation (5G) networks and beyond. These systems assume a large number of distributed access points (APs) using joint coherent transmission to communicate with the users. However, CF mMIMO systems present an increasing computational complexity as the number of users increases. Scalable cell-free CF (SCF) systems have been proposed to face this challenge. Given that the cost-efficient deployment of such large networks requires low-cost transceivers, which are prone to unavoidable hardware imperfections, realistic evaluations of SCF mMIMO systems should take them into account before implementation. Hence, in this work, we focus on the impact of hardware impairments (HWIs) on the SCF mMIMO systems through a general model accounting for both additive and multiplicative impairments. Notably, there is no other work in the literature studying the impact of phase noise (PN) in the local oscillators (LOs) of CF mMIMO systems or in general the impact of any HWIs in SCF mMIMO systems. In particular, we derive upper and lower bounds on the uplink capacity accounting for HWIs. Moreover, we obtain the optimal hardware-aware (HA) partial minimum mean-squared error (PMMSE) combiner. Especially, the lower bound is derived in closed-form using the theory of deterministic equivalents (DEs). Among the interesting findings, we observe that separate LOs (SLOs) outperform a common LO (CLO), and the additive transmit distortion degrades more the performance than the additive receive distortion.
	\end{abstract}
	
	
	\section{Introduction}
	The radical new concept of massive multiple-input-multiple-output (mMIMO), initially proposed by Marzetta \cite{Marzetta2010}, has become a mainstream and mature technology coping with the increasing demand for high data rates in the fifth generation (5G) wireless networks \cite{Marzetta2016,massivemimobook}. In particular, mMIMO systems, which are aimed at serving a smaller number of users than antennas, improving the spectral efficiency (SE) by at least $ 10 \times $ compared to previous generation networks by utilizing simple signal processing and equipping the existing cell-sites with more antennas \cite{massivemimobook}. Notably, the mMIMO paradigm appears under two design approaches related to the antennas deployment: a co-located antenna array as in 5G or a geographically distributed array with a large number of access points (APs) known as cell-free (CF) mMIMO systems. Both approaches exploit coherent processing, where multiple antennas/APs support processing the signal of the same user. This paper considers the second scenario, which is basically a beyond 5G technology.
	
	The CF mMIMO infrastructure was proposed in \cite{Ngo2017} as a promising way to alleviate the cell boundary concept from the cellular architecture. Also, this technology enriches the operating regime with many more single or multiple antenna access points (APs) than user equipments (UEs). Specifically, a large number of APs, connected to a central processing unit (CPU), serves jointly the UEs while benefiting from the advantages of ultra-dense networks and mMIMO systems. Hence, CF mMIMO systems outperform small-cell systems and co-located mMIMO by providing (almost) uniformly good quality of service due to lower average path-loss, increased macro-diversity, some degree of channel hardening, favourable propagation, and suppressed interference. Recently, \cite{Bjoernson2019} suggested a taxonomy with four different implementations
	of CF mMIMO systems and that a centralized implementation of the minimum mean-squared error (MMSE) combining outperforms maximum-ratio
	combining (MRC) with reduced fronthaul signaling. 
	
	Unfortunately, the majority of papers on CF mMIMO assume that all APs are connected to one CPU, being responsible to undertake the coordination and processing of all UEs' signals \cite{Ngo2017,Bjoernson2019,Papazafeiropoulos2020a,Papazafeiropoulos2021,Papazafeiropoulos2021a}. Consequently, the computational burden and fronthaul requirements grow further with the number of UEs and indicate that the primitive model of CF mMIMO systems is not scalable. Interestingly, \cite{Buzzi2017a,Buzzi2019,Bjoernson2019a} suggested a user-centric approach, where each UE is not served by all APs but a subset providing the best channel conditions. Thus, every UE connects to different subsets, or equivalently, each AP cooperates with different APs to serve different UEs. However, in practice, the system characteristics such as the channels and UE locations vary with time. Taking this into account, in \cite{Bjoernson2019a}, the dynamic cooperation clustering (DCC) from \cite{Bjornson2011} has been applied to achieve scalable CF (SCF) mMIMO systems.

	The prior literature of mMIMO systems assumed perfect hardware, however, hardware impairments (HWIs) are inevitable in practical applications. A cost-efficient implementation of a large number of antenna elements in co-located or distributed layouts, which would consist of low-quality components with severe HWIs, should also make use of handset-grade components. Such components are particularly susceptible to HWIs such as in-phase/quadrature-phase (I/Q)-imbalance~\cite{Qi2010}, oscillator phase noise (PN)~\cite{Pitarokoilis2015,Papazafeiropoulos2016}, and high power amplifier non-linearities~\cite{Qi2012}. There are practical design tradeoffs to adhere to, for example, a power amplifier that is efficient in terms of power-added efficiency might create strong nonlinear distortion, and vice-versa. Even if calibration schemes and compensation algorithms are utilized at the transmitter and receiver, respectively, a certain amount of distortions, known as residual HWIs, remains and can be categorized into additive and multiplicative distortions. The model of additive HWIs includes additive Gaussian noises at both the transmitter and receiver, expressing the aggregate effect of many impairments, and has been grounded based on its analytical tractability and experimental validation \cite{Schenk2008,Studer2010}. The second category includes the PN, which is expressed in terms of the phase drifts emerging from the local oscillators (LOs) and appearing as multiplicative factors to the channels \cite{Pitarokoilis2015}. Notably, the PN accumulates within the channel coherence time. Despite the importance of HWIs only a few works have addressed their impact in the case of CF mMIMO systems \cite{Zhang2018,Zhang2019,Masoumi2019,Hu2019,Zheng2020} while none work has accounted for the effect of PN. In other words, these works focused only on the additive distortions, while no phase noise and amplified thermal noise (ATN) are considered. Moreover, \cite{Bjornson2015}, concerning mMIMO with distributed arrays, which is essentially the same as CF mMIMO systems, has studied additive HWIs and PN but perfect hardware was assumed at the user devices, and no scalability concerns were considered. In summary, there is no prior work investigating the impact of HWIs on SCF mMIMO systems.

	\begin{table}
		\begin{center}
			\caption{Comparison between current and existing works}
			\begin{tabular}{ |>{\centering\arraybackslash}p{1.8cm} |>{\centering\arraybackslash}p{1.8cm}| >{\centering\arraybackslash}p{1.8cm} |>{\centering\arraybackslash}p{1.7cm}| }
				\hline
				\vspace{0.001cm}	Papers & Additive distortions & Multiplicative 	distortion	& 	\vspace{0.001cm} All 3 HWIs \\ \hline
				mMIMO & [36]& [14], [15], [31] & [23], [33], [35]\\ \hline
				CF mMIMO & \cite{Zhang2018,Zhang2019,Masoumi2019,Hu2019,Zheng2020}& Current work&Current work \\\hline
				SCF	mMIMO &Current work & Current work &Current work\\ 
				\hline
			\end{tabular}
			\label{table:factors}
		\end{center}
	\end{table}

	\subsection{Motivation}
	This paper is incited by the following observations:~1) HWIs are unavoidable in practical wireless communications and since they are commonly neglected in theory, there is a gap between theory and practice,~2) mMIMO is a more attractive technology if its implementation is based on cheap hardware, which are more susceptible to HWIs,~3) CF mMIMO systems is a promising architecture but its feasibility depends on its scalability in large networks with increasing number of UEs,~4) recently, it was shown that CF mMIMO systems outperform small cells and cellular systems if a centralized version of MMSE decoding is applied, which also reduces the fronthaul signaling compared to distributed systems,~5) previous works in CF mMIMO systems, assuming MMSE decoding, did not obtain any closed-form expressions but relied only on Monte-Carlo simulations, ~6) channel state information (CSI), being of paramount importance in CF mMIMO systems, is drastically affected by the presence of HWIs, and~7) there is no prior work investigating the impact of HWIs in emerging SCF mMIMO systems. Actually, the PN has not even been studied in the case of standard CF mMIMO systems, while a few existing works focused only on a single type of HWIs \cite{Zhang2018,Zhang2019,Masoumi2019,Hu2019,Zheng2020}. Notably, since HWIs are time-varying, they affect further SCF systems by means of the DCC being time-dependent.
	
	These observations suggest that there is a great practical interest in investigating the HWIs in SCF mMIMO systems with MMSE combining in terms of deriving closed-form expressions that accounts for both additive and multiplicative HWIs. 
	
	\subsection{Contributions }
	The main contributions are summarized as follows.
	\begin{itemize}
		\item Contrary to the existing works \cite{Zhang2018,Zhang2019,Masoumi2019,Hu2019,Zheng2020}, which have studied only a single kind of HWIs in CF mMIMO systems, we introduce a general model for HWIs including both additive and multiplicative (PN) distortions for the more practical architecture of SCF mMIMO systems. The introduction of the HWIs requires substantial changes to the analysis of SCF mMIMO systems. There is no prior work addressing PN at CF mMIMO systems or in general any HWIs at SCF mMIMO systems.\footnote{Note that this work extends substantially our conference paper~\cite{Papazafeiropoulos2020}. Among others, it provides an upper bound on the uplink capacity, and it obtains the hardware-aware (HA) MMSE combiner by minimizing the MSE. In particular, contrary to \cite{Papazafeiropoulos2020} relying on a hardware unaware (HU) MMSE combiner, it presents the achievable SE based on the HA version of the MMSE combiner and elaborates on their comparison. } Table \ref{table:factors} provides a comparison between the current and existing works and clearly demonstrates the open problems we address.
		\item Previous works in CF mMIMO systems that applied MMSE combining presented only simulations while, in this work, we employ MMSE combining with imperfect CSI and obtain closed-form expressions using the deterministic equivalent (DE) analysis. Notably, no other work has provided DE expressions for SCF mMIMO systems. %Monte-Carlo simulations verify the analytical results.
		\item We derive an upper bound on the uplink capacity and the corresponding optimal MMSE combiner. Next, we obtain the hardware-aware (HA)-MMSE combiner by minimizing the sum mean
		square error (MSE), and derive the corresponding DE lower bound on the uplink capacity. We also consider several decoders such as the conventional MRC and hardware-unaware (HU)-MMSE decoders.
		\item We shed light on the impact of HWIs on the uplink achievable SE of SCF mMIMO systems. We show how the PN degrades the system performance. Specifically, it is shown that separate local oscillators (SLOs) and common LO (CLO) designs behave differently, i.e., the achievable SE with SLOs appears superior against the CLO setting. Moreover, we quantify the degradation of the system due to additive HWIs, where the transmit distortion has a more severe impact than the received one. In addition, we show the increase in the gap between SLOs and CLO deployments as the number of APs increases.
	\end{itemize}
	
	%The remainder of this paper is structured as follows. 
	\textit{Paper outline:} Section~\ref{System} presents the system model of an SCF mMIMO system in terms of fundamental design parameters. In addition, a description of the HWIs is provided. In Section~\ref{Channel6Estimation}, the scalable channel estimation with HWIs is described. 
	%In the same section, the criterion for making the cell selection is presented, and we obtain the association region defining the area that a typical user is associated with each tier. 
	In Section~\ref{PerformanceAnalysis}, we elaborate on the uplink data transmission, and obtain an upper and a lower bound on the uplink capacity as well as the HA-PMMSE decoder. 
	Section~\ref{Deterministic} presents the DE of the lower bound on the uplink capacity with HWIs. The numerical results are placed in Section~\ref{results}, while Section~\ref{Conclusion} summarises the paper.
	
	\textit{Notation:} Vectors and matrices are denoted by boldface lower and upper case symbols. The symbols $(\cdot)^\T$, $(\cdot)^\H$, and $(\cdot)^\dagger$ express the transpose, Hermitian transpose, and pseudo-inverse operators, respectively. The expectation and variance operators are denoted by $\EE\left[\cdot\right]$ and $\mathrm{Var}\left[\cdot\right]$, respectively. The notations $\mathbb{C}^{M \times 1}$ and $\mathbb{C}^{M\times L}$ refer to complex $M$-dimensional vectors and $M\times L$ matrices, respectively. Also, the notation $\xrightarrow[ M \rightarrow \infty]{\mbox{a.s.}}$ denotes almost sure convergence as $ M \rightarrow \infty $. The notation $a_n\asymp b_n$ is equivalent to the relation $a_n - b_n \xrightarrow[ M \rightarrow \infty]{\mbox{a.s.}} 0$, where $a_n$ and $b_n$ are two infinite sequences. Finally, $\bb \sim \cC\cN{(\b0,\mathbf{\Sigma})}$ represents a circularly symmetric complex Gaussian variable with zero mean and covariance matrix $\mathbf{\Sigma}$.
	
	\section{System Model}\label{System}
	We consider a CF network architecture including $M$ APs and $K$ UEs randomly distributed over a geographic area with $ M $ and $ K $ being large. The APs, each equipped with $L$ antennas and fully digital transceiver chains, are connected to central processing units through backhaul links in an arbitrary way~\cite{Perlman2015,Interdonato2019}. Note that we focus on sub-6 GHz bands, thus hybrid beamforming solutions are not considered. Each UE is equipped with a single antenna. All communications take place in the same time-frequency resources. In particular, coherent joint transmission and reception are enabled among the APs. While $M\gg K$ is typical in CF MIMO, there is no formal assumption of that kind in this paper.
	
	\subsection{Channel Model}\label{ChannelModel} 
	We consider a time-varying narrowband channel that is divided into coherence blocks. We employ the standard block fading model %time-division-duplex (TDD) protocol,
	where the length of each coherence interval/block, being $\tau_{\mathrm{c}}$ channel uses, is equal to the product of the coherence bandwidth $B_{\mathrm{c}}$ in $\mathrm{Hz}$ and the coherence time $T_{\mathrm{c}}$ in $\mathrm{s}$~\cite{massivemimobook}. 		We consider 
	$\tau$ channel uses per block for uplink training to enable channel estimation, and $\tau_{\mathrm{u}}=\tau_{\mathrm{c}}-\tau$ channel uses are used for uplink data transmission. The channel use is denoted by $ n \in \{ 1, \ldots, \tau_{\mathrm{c}}\} $.
	% , i.e., $\tau_{\mathrm{c}}=\tau_{\mathrm{p}}+\tau_{\mathrm{u}}+\tau_{\mathrm{d}}$ .
	% of the $n$th symbol 
	
	During the transmission in each coherence block, the channel vector between the $m$th AP and the $k$th UE, denoted by $\bh_{mk} \in\mathbb{C}^{L}$, is fixed and exhibits flat-fading. %is assumed constant for a symbol period but it may vary slowly from symbol to symbol\footnote{The symbol duration is assumed smaller or equal to the coherence time of all the users}. 
	The random channel vector is expressed by means of an independent correlated Rayleigh fading distribution as
	\begin{align}
		\!\bh_{mk}\sim \mathcal{CN}\!\!\left( \b0, \bR_{mk} \right)\!, ~~m\!=\!1,\ldots,M~\mathrm{and}~k\!=\!1,\ldots,K
	\end{align}
	where the complex Gaussian distribution models the small-scale fading and $ \bR_{mk} \in \mathbb{C}^{L\times L}$ is a deterministic positive semi-definite correlation matrix that describes the large-scale propagation effects such as shadowing and geometric pathloss with large-scale fading coefficient $ \beta_{mk}={\tiny \tr\left(\bR_{mk}\right)}/L $. Notably, we assume that the set $ \left\{ \bR_{mk} \right\} $ is known by the network elements whenever needed with the help of practical methods (see e.g., \cite{Neumann2018,Upadhya2018}).
	% \footnote{The matrix $\bR_{k}$ is independent of the time index, $n$ because we assume that its representing effects vary with time in a slower pace than the coherence time $\tau_{\mathrm{c}}$~\cite{Truong2013,Papazafeiropoulos2014,Papazafeiropoulos2014WCNC}.}

	\subsection{Basic Scalability Guidelines}\label{Scalability} 
	The original implementation of CF mMIMO systems where all APs serve all UEs is impractical in terms of cost and complexity in large deployments with large $ K $. To make the technology feasible and attractive, it has to be scalable, in particular, regarding the number of UEs~\cite{Bjoernson2019a}. Hence, the authors examined the scalability issue and provided a corresponding definition by letting $ K \to \infty $. Specifically, they defined a CF mMIMO network as scalable when:\textit{ i)} the signal processing for channel estimation,\textit{ ii)} the signal processing for data reception and transmission,\textit{ iii)} the fronthaul signaling for data and CSI sharing, and \textit{iv)} the power control optimization have finite complexity and resource requirements per AP when $ K \to \infty $.\footnote{The intention with this definition is not that the network will serve infinitely many UEs, which is impractical but to make a design where the resources and complexity per AP remain limited even if $K$ is large but finite.} Notably, the complexity properties of SCF, and their advantages over conventional CF systems, have been thoroughly analyzed in \cite{Demir2021}, and apply also to the setup considered in this paper.
	
	Under these conditions, the network can serve more UEs and deploy more APs, without having to upgrade previously deployed APs. 
	To this end, each AP serves a set of UEs with constant cardinality as the total number of UEs $ K \to \infty $. In order to satisfy the scalability conditions above, the set of UEs served by AP $ m $ is defined as~\cite{Bjoernson2019a}
	\begin{align}
		\mathcal{D}_{m}=\left\{i :~\tr\left( \bD_{mi}\right )\ge 1,~i\in \left\{1,\ldots,K\right\}\right\},
	\end{align} 
	where $ \bD_{mi} \in \mathbb{C}^{L\times L} $ denotes a diagonal matrix defining which APs and UEs communicate with each other according to the DCC framework \cite{Bjornson2011}. These matrices enable a unified analysis with the original CF mMIMO system being one of the many architectures that could be described by this framework.

	\subsection{Hardware Impairments}\label{HardwareImpairments} 
	Many practical systems use the direct-conversion architecture where the transmitter contains digital-to-analog converters (ADCs), low-pass filters, I/Q mixers, and high-power amplifiers, while the receiver contains filters, low-noise amplifiers, I/Q mixers, and ADCs.\footnote{Especially, \cite{Joung2014} is an informative survey article on power amplifiers, representing one of the most energy-consuming components in wireless systems.} Although the majority of papers in the CF mMIMO literature assumes ideal transceiver hardware, physical transceiver implementations such as the converters and the oscillators undergo to HWIs and inevitably distort the signal. Despite the application of compensation and mitigation algorithms, the HWIs are not completely eliminated. The main reason for this is the limited modeling accuracy of various system parameters. The choice of specific hardware quality is crucial for the overall cost and power consumption. In particular, attractive solutions for large deployments such as the CF mMIMO technology under study should take into account the use of cheap hardware that results in low cost and power consumption. Otherwise, the deployment of this technology would be prohibitive in terms of cost and energy efficiency. Hence, it is of paramount importance to incorporate in the analysis of CF mMIMO systems the impact of the hardware constraints.
	
	%From a system viewpoint, the aggregate effect of the overall residual impairments is of high importance and not the impact of the individual hardware components. In this direction, 
	We are purposely not limiting this work to a single hardware implementation but rather taking a holistic approach, where we characterize three main categories of residual HWIs that appear in all practical implementations. These HWIs can be classified
	% with regards to the accompanying distortions 
	into $1)$ a multiplicative distortion due to phase shift of the signal, $2)$ additive distortions at both the transmitter and the receiver with powers proportional to the transmit and the received signal, respectively, and $3)$ a channel-independent ATN.

	\subsubsection{Multiplicative distortion}\label{multiplicativedistortion} 
	The multiplicative distortion of the channel expresses time-dependent random phase-drifts known as PN, and it is induced during the up-conversion of the baseband
	signal to passband and vice-versa by multiplying the signal with the LO's output. Based on standard references concerning the PN~\cite{Pitarokoilis2015}, this distortion at the $n$th channel use can be expressed by a discrete-time independent Wiener process.\footnote{The Wiener process model assumes independent innovations at every time instance, which might be either created by nature or by residual phase noise after some compensation algorithm having been applied. Our channel estimator, provided in Section~\ref{Channel6Estimation}, tries to compensate for the phase noise in the sense of estimating the effective channel at time $ n $, based on pilot signals obtained at a different time. By utilizing the structure of the phase noise in the estimator and then using the estimate for receive combining, we compensate for the phase noise.} The PNs at the LOs of the $j$th antenna of the $m$th AP and $k$th UE, respectively, at the $n$th channel use are modeled as
	\begin{align}
		\phi^{j}_{m,n}&=\phi^{j}_{m,n-1}+\delta^{\phi}_{n}\label{phaseNoiseAP}\\
		\varphi_{k,n}&=\varphi_{k,n-1}+\delta^{\varphi}_{n},\label{phaseNoiseuser}
	\end{align}
	where $\delta^{\phi}_{n}\sim \cN(0,\sigma_{\phi}^{2}) $ and $\delta^{\varphi}_{n}\sim \cN(0,\sigma_{\varphi}^{2})$~\cite{Petrovic2007,Pitarokoilis2015,Krishnan2015}. The increment variance of the PN process can be expressed by \cite{Petrovic2007}
	\begin{align}
		\sigma_{i}^{2}=4\pi^{2}f_{\mathrm{c}}^{2} c_{i}T_{\mathrm{s}},~~~~ i=\phi, \varphi\label{PN1}
	\end{align}
	where $T_{\mathrm{s}}$, $c_{{i}}$, and $f_{\mathrm{c}}$ are the symbol interval, a constant dependent on the oscillator, and the carrier frequency, respectively. Note that the total PN process from the $k$th user and the $j$th antenna of the $m$th AP is $\theta_{mk,n}^{(j)}= \phi^{j}_{m,n}+\varphi_{k,n}$ while $\bTheta_{mk,n}\!=\!\mathrm{diag}\!\left\{ e^{i \theta_{mk,n}^{(1)}}, \ldots, e^{i\theta_{mk,n}^{(L)}} \!\right\}$ denotes the corresponding total PN matrix at the $n$th channel use. In fact, this generic matrix accounts for the general scenario (non-synchronous operation), where we have SLOs at each antenna justifying the independence among the PN processes~\cite{Pitarokoilis2015,Krishnan2015}. In the special case of the synchronous operation, all PN processes $\theta^{\left( j \right)}_{mk,n}$ are identical for all $ j=1,\ldots,L$. The matrix degenerates to $\Theta_{mk,n}\!=\! e^{i \theta_{mk,n}^{(1)}}\Id_{L}$ when there is just one CLO connected to all antennas of an AP. In our analysis, we focus on both SLOs and CLO scenarios.

	%We consider a general setup,
	%, denoted as the general oscillator (GO) setup
	%where each AP is implemented with $ L $ antennas and each antenna is connected to a free-running oscillator. Specifically, the matrix 
	%$\bTheta_{mk,n}\!=\!\mathrm{diag}\!\left\{ e^{i \theta_{mk,n}^{(1)}}\bf{1}_{1 \times M/M_{\mathrm{o}}}^{\T}, \ldots, e^{i\theta_{mk,n}^{(M_{\mathrm{o}})}} \bf{1}_{1 \times M/M_{\mathrm{o}}}^{\T}\!\right\}$, 
	%$\bTheta_{mk,n}\!=\!\mathrm{diag}\!\left\{ e^{i \theta_{mk,n}^{(1)}}, \ldots, e^{i\theta_{mk,n}^{(L)}} \!\right\}$, 
	%where $ \bf{1}_{1 \times M/M_{\mathrm{o}}}^{\T} $ denotes an all-one vector of length $ M/M_{\mathrm{o}}$, 
	%expresses the total PN at the $n$th channel use, where $\theta^{\left( j \right)}_{mk,n}, j=1,\ldots,L$ are independent PN processes with $\theta^{\left( j \right)}_{mk,n}$ being the PN process from the $j$th antenna of the $m$th AP and the single-antenna from UE $k$. 

	Henceforth, we will use the following lemma providing the expectation of the complex exponential of the phase drift $ \Delta\phi=\phi_{n_{2}}-\phi_{n_{1}} $, where $ \phi_{n} $ is a general PN process with increment variance $\sigma_{\phi}^{2} $. Also, we denote $\Delta{n}=n_{2}-n_{1} $ where $ n_{1},n_{2} $ are two different time slots. The phase drift is 
	a zero-mean Gaussian variable with a variance proportional to $ \Delta{n} $ distributed as $ 
	\Delta\phi \sim \mathcal{N}\left (0, \sigma_{{\phi}}^{2}\Delta{n}\right)$, i.e., the variance increases with time. 
	\begin{lemma}\label{LemmaPN}
		The expectation of the complex exponential of the phase drift $ \Delta\phi_{n} $ occured during $ \Delta{n}$ is
		\begin{align}
			\EE\left\{e^{-j \Delta\phi}\right\}=e^{-\frac{\sigma_{\phi}^{2}}{2}\Delta{n}}.
		\end{align}
		\begin{proof}
			The proof is straightforward since this expectation is basically the characteristic function of the Gaussian variable $ \Delta\phi $ having zero mean and variance $ \sigma_{\phi}^{2} $.
		\end{proof}
		
	\end{lemma}

	\subsection{Additive distortions}\label{Additivedistortion} 
	In real systems, 
	%despite the mitigation schemes at both the transmitter and receiver side,
	unavoidable HWIs emerge due to the imperfect compensation of the nonlinearities in the power amplifiers in the transmitter, the quantization noise in the ADCs at the receiver, the I/Q imbalance, etc. \cite{Schenk2008,Studer2010}. As a result, both the transmit and receive signals are distorted during the transmission and reception processing, respectively. For example, at the transmitter side, a nonlinear mismatch appears between the generated signal and the signal intended to be transmitted.\footnote{At the receiver, the signal is quantized to be described by a limited number of bits. We assume the fronthaul connections are able to carry these bits without incurring additional distortion.}

	Measurement results have shown that the conditional additive distortion noise for the $i$th link, given the channel realizations,
	is modeled as Gaussian distributed with average power proportional to the average signal power \cite{Studer2010}. The Gaussianity is based on the consideration that this distortion results because of the aggregate contribution of many impairments. Hence, the impairments at different antenna branches are modeled as mutually uncorrelated Gaussian
	random variables. Moreover, this distortion depends on the time since it takes a new realization for each data signal being time-dependent itself.
	%\subsubsection{Downlink}
	%The transmit and receive distortions are written as
	%\begin{align}
	%\deltav_{\mathrm{t},n}^{m}&\sim \cC\cN\left( \b0,\bm \Lambda^{m}_{n} \right)\label{eta_tD} \\
	%\delta_{\mathrm{r},n}^{k}&\sim \cC\cN \left( 0, \Upsilon^{k}_{n} \right)\label{eta_rD},
	%\end{align}
	%where $\bm \Lambda^{m}_{n}= \kappa_{\mathrm{t}_\mathrm{AP}}^{2}\mathrm{diag}\left( q_{1}[n],\ldots,q_{L}[n] \right)$ and 
	%$ \Upsilon^{k}_{n} =\kappa_{\mathrm{r}_{\mathrm{UE}}}^{2}\sum_{i=1}^{M} \bh_{ik}^{\H}\bQ_{i}[n]\bh_{ik} $ with the proportionality parameters $\kappa_{\mathrm{t}_{\mathrm{AP}}}^{2}$ and $\kappa_{\mathrm{r}_{\mathrm{UE}}}^{2}$ describing the severity of the residual impairments at the transmitter and the receiver side Moreover, $\bQ_{i}[n]$ denotes the transmit covariance matrix at time $n$ with diagonal elements $q_{{i}_1}[n],\ldots,q_{T_{i}}[n]$. Notably, the transmit distortion does not depend on the receiver while the receive distortion depends on both the transmitter and the receiver as expected. 
	%\subsubsection{Uplink}
	Especially, in the uplink case, the additive transceiver distortions are expressed in terms of conditional distributions as
	\begin{align}
		\delta_{\mathrm{t},n}^{k}&\sim \cC\cN\left( 0, \Lambda^{k}_{n} \right)\label{eta_tU}, \\
		\deltav_{\mathrm{r},n}^{m}&\sim \cC\cN \left( \b0,\bm \Upsilon^{m}_{n} \right)\label{eta_rU}\!,
	\end{align}
	where $ \Lambda^{k}_{n}= \kappa_{\mathrm{t}}^{2}\rho_{k}$ and 
	$\bm \Upsilon^{{m}}_{n} =\kappa_{\mathrm{r}_{m}}^{2}\sum_{i=1}^{K}\rho_{i} \mathrm{diag}\big( |h_{mi}^{\left(1\right)}|^{2},$ $\ldots,|h_{mi}^{\left(L\right)}|^{2} \big) $ with $ \rho_{k} $ being the transmit power from UE $ k $ while the proportionality parameters $\kappa_{\mathrm{t}}^{2}$ and $\kappa_{\mathrm{r}_{m}}^{2}$ describe the severity of the residual impairments at the transmitter and the receiver side, respectively. In practice, the error vector magnitude (EVM) of the hardware can be measured and then mapped into the parameter values above~\cite{Holma2011}.

	\subsection{ATN}\label{Amplifiedthermalnoise}
	Certain components such as the low noise amplifiers and mixers induce an amplification of the thermal noise denoted as $ \bxi_{m,n}$ that manifests as an increase of the variance of the thermal noise. In particular, this amplified receiver noise is modeled as Gaussian distributed with zero-mean and variance $\xi_{m} \Id_{L}$ with $\xi_{m}\ge \sigma^{2}$ where $\sigma^{2}$ 
	is the corresponding parameter of the actual thermal noise~\cite{Bjornson2015}. In addition, it is worthwhile to mention that it is independent
	of the UE channels. Also, note that the ATN is time-dependent taking different random realizations over time because it consists of samples obtained from a white noise process that has passed by
	some amplified ``filter''~\cite{Papazafeiropoulos2017a}.
	\begin{remark}
		If we set $ \sigma_{\phi}^{2}=\sigma_{\varphi}^{2}=\kappa_{\mathrm{t}}=\kappa_{\mathrm{r}}=0 $ and $ \xi_{m}= \sigma^{2} $, the analytical results, obtained below, can describe the unrealistic case of no hardware impairments. Notably, the analytical results do not reduce to any known expressions without hardware impairments since these expressions have not been obtained by previous works. The previous works relied only on Monte-Carlo simulations, while we provide a novel closed-form SE expression in Theorem 1. However, the proposed signal models and the combiners reduce to the signal models and the combiners in \cite{Bjoernson2019a}.
	\end{remark}
	\section{Scalable Channel Estimation with Hardware Impairments}\label{Channel6Estimation}
	
	Since the channels vary independently between coherence blocks, the APs need to estimate them once per block. In most works concerning CF mMIMO, the APs estimate the channel based on pilot signals transmitted in the uplink training phase~\cite{Ngo2017}. 
	% In such case, both the uplink and downlink channels are known due to the property of reciprocity.
	% Note that since the duration of the coherence interval is limited, there are not enough mutually orthogonal training sequences for all UEs.
	%We assume that the training phase takes place at time $0$. Moreover, w
	%We denote by $ \mathcal{S} \subset \set{1, \ldots, K}$ the subset of the 
	To this end, the first $\tau < K$ channel uses of every coherence block are used for pilot transmission.
	Each UE transmits a $\tau$-length
	mutually orthogonal training sequences, i.e., $\bm \omega_{k}= \left[\omega_{k,1},\ldots,\omega_{k,\tau} \right]^{\T}\in \bbC^{\tau \times 1}$ with $\rho_{\rp}=|\omega_{k,n}|^{2} ,\forall k,n$, with $\rho_{\rp}$ being the common average transmit power per UE. 	Note that the subscript $ \mathrm{p} $ corresponds to variables describing the training phase. Since the dimension of the pilot sequences can be smaller than the number of UEs, some sequences will be shared by multiple UEs, leading to pilot contamination. Also, the system model supports arbitrary pilots.
	%while the sequences among different UEs are mutually %orthogonal, 
	% $\sqrt{\tau}\bpsi_{i} \in \bbC^{ \tau \times 1}$ with $\|\bpsi_{k} \|^{2}= 1 $, 
	%	in order that each UE has a unique pilot sequence not shared with any other UE. 
	% The corresponding concatenated matrix is written as $\bPsi = [\bpsi_{1}; \cdots; \bpsi_{K}] \in \bbC^{K \times \tau}$ with $\|\bpsi_{k} \|^{2}= 1 $. 
	Moreover, according to~\cite[Assumption~1]{Bjoernson2019a}, each AP serves at most one UE per pilot sequence, which implies that $ |\mathcal{D}_{m}|\le \tau$ and 
	\begin{align}
		\bD_{mk}&= \left\{\begin{array}{ll}
			\Id_{L}&~\mathrm{if}~k \in \mathcal{D}_{m} \\
			\b0&~\mathrm{if}~k\notin \mathcal{D}_{m}. \end{array}\label{UC}
		\right.
	\end{align}
	Notably, the independence of $ \tau $ from $ K $ agrees with the requirement for scalability since the length of each coherence block, used for channel estimation does not grow with the number of UEs.
	
	%By assuming that the channel estimation takes place at time $0$,

	During the training phase, the received signal by the $ m $th AP at channel use $ n $, coming all UEs and including the HWIs, is given by
	\begin{align}
		\by_{m,n}= & \sum_{i =1}^K \bTheta_{mi,n} \bh_{mi}\left(\omega_{i,n} + \delta_{\mathrm{t},n}^{i} \right)+\deltav_{\mathrm{r},n}^{m}+
		\bxi_{m,n},\label{eq:Ypt}
	\end{align}
	where we have incorporated the HWIs, i.e., $\bTheta_{mi,n}=\mathrm{diag}\left\{ e^{j \theta_{i,n}^{(1)}}, \ldots, e^{j \theta_{i,n}^{(L)}} \right\}$ is the phase noise because of the LOs of the $ m $th AP and UE $k$ at channel use $n$, $\delta_{\mathrm{t},n}^{i} \sim\mathcal{CN}\left(0,\kappa_{\mathrm{t}}^{2}\rho_{\rp}\right) $ is the additive transmit distortion, $\deltav_{\mathrm{r},n}^{m}\sim \cC\cN \left(\b0,\bm \Upsilon^{m}_{n} \right)$ is the additive receive distortion with $\bm \Upsilon^{{m}}_{n} =\kappa_{\mathrm{r}_{m}}^{2}\sum_{i=1}^{ \mathcal{K}}\rho_{i} \mathrm{diag}\left( |h_{mi}^{\left(1\right)}|^{2},\ldots,|h_{mi}^{\left(L\right)}|^{2} \right) $, and $ \bxi_{m,n}$ is the spatially amplified Gaussian thermal noise matrix at the $ m $th AP. 
	
	%
	%By correlating the received signal with the corresponding pilot sequence $\bpsi_k$ from UE $ k $, the $ m $th AP estimates the channel from the $ k $th UE as
	%\begin{align}
	%\tilde{\by}_{mk,0}^{\rp} &= \bY_{m,0}^{\rp}\bpsi_k \nn\\
	%&=\sum_{i \in \mathcal{S}} \sqrt{\tau p_{i}}\bTheta_{mi,0} \bh_{mi}\left(\bpsi_{i}^{\H}\bpsi_k +\tilde{\delta}_{\mathrm{pt},0}^{{i} } \right)+ \tilde{\deltav}_{\rp \mathrm{r},0}^{m}+ \tilde{\bxi}_{\rp m,0},\label{eq:Ypt3}
	%\end{align}
	%where $\tilde{\delta}_{\mathrm{pt},0}^{{i} }=\deltav_{\mathrm{pt},0}^{{i~{\scriptscriptstyle\mathsf{H}}} }\bpsi_{k}\sim \cC\cN(0,{\tau} \Lambda^{i}_{\rp,n})$, $ \tilde{\deltav}_{\rp \mathrm{r},0}^{m}= \bDelta_{\rp \mathrm{r},0}^{m}\bpsi_{k}\sim \cC\cN(\b0,\tau \Upsilon^{{m}}_{n}\bI_{L})$, and $\tilde{\bxi}_{\rp m,0}=\Xim_{m,0}^{\rp}\bpsi_{k}\sim \cC\cN(\b0,\tau \xi_{n}^{{m}}\bI_{L})$. 
	%Application of the MMSE estimation method (see~\cite{Kay}) results in the estimated channel at the $ m $th AP.
	\begin{remark}
		The received signal depends on both multiplicative and additive HWIs.\footnote{For the sake of a better presentation of the transceiver HWIs, herein, we provide the ideal uplink transmission of a SCF mMIMO system without any imperfections during training. Thus, in such case, the received signal at the $m$th AP at a given channel use $n$ is expressed as
			\begin{align}
				\by_{m,n}= & \sum_{i=1}^K \bh_{mi}\omega_{i,n} +
				\bw_{m,n},\label{eq:Ypt5}
			\end{align}
			where $ \bw_{m,n}\sim\mathcal{CN}\left(\b0,\sigma^{2}\Id_{L}\right) $.
		} 
	\end{remark}

	\begin{proposition}
		The linear minimum mean-square error (LMMSE) estimate of the effective channel of UE $k$ $ \bh_{mk,n} =\bTheta_{mk,n} \bh_{mk} $ at any channel
		use $ n \in \{ 1, \ldots, \tau_{\mathrm{c}}\} $ is given by
		\begin{align}
			\hat{\bh}_{mk,n}= \left( \bm \omega_{k}^{\H}\bm \Lambda_{k,n}\otimes \bR_{mk}\right)\bQ_{m}^{-1}\bm \psi_{m}\label{EstimatedChannel} ,
		\end{align}
		where 
		\begin{align}
			\bm \psi_{m}&=\left[\by_{m,1}^{\T}, \ldots, \by_{m,\tau}^{\T}\right]^{\T},\\
			\bm \Lambda_{k,n}&=\mathrm{diag}\!\left\{ e^{-\frac{\sigma_{\varphi}^{2}+\sigma_{\phi}^{2}}{2}|n-1|}, \ldots, e^{-\frac{\sigma_{\varphi}^{2}+\sigma_{\phi}^{2}}{2}|n-\tau|} \!\right\},\\
			\bQ_{m}&= \sum_{j=1 }^K\bX_{mj}\otimes\bR_{mj}+\xi_{m}\Id_{\tau L},\label{sum17}\\
			\bX_{mj}&= \tilde{\bX}_{j}+\left(\kappa_{\mathrm{t}_\mathrm{UE}}^{2}+\kappa_{\mathrm{r}_{m}}^{2}\right)\bL_{|\bm \omega_{j}|^{2}},\\
			\bL_{|\bm \omega_{j}|^{2}}&= \mathrm{diag}\left( |\omega_{j,1}|^{2},\ldots,|\omega_{j,\tau}|^{2} \right),\\
			\left[ \tilde{\bX}_{j}\right] _{u,v}&= \omega_{j,u}\omega_{j,v}^{*} e^{-\frac{\sigma_{\varphi}^{2}+\sigma_{\phi}^{2}}{2}|u-v|}.
		\end{align}
	\end{proposition}
	\begin{proof}
		We start by the definition, providing the estimated channel vector, given by
		\begin{align}
			\hat{\bh}_{mk,n}&=\EE\left[{\bh}_{mk,n}\bm \psi^{\H}_{m} \right]\left( \EE\left[\bm \psi_{m} \bm \psi^{\H}_{m} \right] \right)^{-1}\bm \psi_{m}.\label{est1}
		\end{align}
		The computation of each term in \eqref{est1} follows the same steps with Theorem 1 in~\cite{Bjornson2015} but it is more general since the proposed model includes the transmit HWI $ \delta_{\mathrm{t},n}^{k}$ and requires some extra algebraic manipulations. 
		%Also, this estimated channel accounts for the scalability design since the summation in \eqref{sum17} is over the sum of participating UEs and not all UEs. %The proof is omitted due to limited space.
	\end{proof}

	Based on the property of orthogonality of LMMSE estimation, the current channel at the $ n $th channel use is given by 
	\begin{align}
		\bh_{mk,n}=\hat{\bh}_{mk,n}+\tilde{\bh}_{mk,n},\label{current} \end{align}
	where $\hat{\bh}_{mk,n}$ and $ \tilde{\bh}_{mk,n} $ have zero mean and variances $\bPhi_{mk}=\left( \bm \omega_{k}^{\H}\bm \Lambda_{k,n}\otimes \bR_{mk}\right)\bm \bQ_{m}^{-1}\left( \bm \Lambda_{k,n}^{\H}\bm \omega_{k}\otimes \bR_{mk}\right)$ and $ \tilde{\bR}_{mk}=\bR_{mk}-\bPhi_{mk}$, respectively.

	%\begin{proposition}\label{LMMSE} 
	%	The HWI aware LMMSE estimator of $\bh_{mk,0}$, obtained during the training phase, is
	%	\begin{align}\label{estimatedChannel}
	%	\hat{\bh}_{mk,0}=\sqrt{\tau p_{k}}\bR_{mk} \bZ_{m}^{-1}\tilde{\by}_{mk,0}^{\rp},
	%	\end{align}
	%	where $\tilde{\by}_{mk,0}^{\rp}$ is the noisy observation of the effective channel from UE $k$ belonging to the subset $ \mathcal{S} $ to the $ m $th AP, and $p_{\mathrm{p}}=\tau p_{\mathrm{u}}$ with $p_{\mathrm{u}}$ is the power per user in the uplink data transmission phase while $ \bZ_{m} $ is given by
	%	\begin{align}
	%\bZ_{m} =	\sum_{i \in \mathcal{S}}\tau p_{i}\left (\bR_{mi}\left ( \|\bpsi_{i}^{\H}\bpsi_k\|^{2}+\kappa_{\mathrm{t}_i}^{2}\right )+\kappa_{\mathrm{r}_{{m}}}^{2}p_{i}\mathrm{diag}\left( \bR_{mi}^{\left (1\right )},\ldots,\bR_{mi}^{\left (L\right )} \right)\right )+\xi_{n}^{{m}} \Id_{L}.
	%	\end{align}
	%\end{proposition}
	%\proof: The proof of Proposition~\ref{LMMSE} is given in Appendix~\ref{proposition1}.\endproof
	%
	%Based on the orthogonality principle, the channel can be decomposed as
	%\begin{align}
	%\bh_{mk,0} = \hat{\bh}_{mk,0} + \tilde{\bee}_{mk,0},\label{eq:MMSEorthogonality}
	%\end{align}
	%where $\hat{\bh}_{mk,0}$ is distributed as $\cC\cN\left(\b0,\bPhi_{mk}\right)$ with $\bPhi_{mk}=\tau p_{k} \bR_{mk}\bZ_{m}^{-1}\bR_{mk}$, and $\tilde{\bee}_{mk,0} \sim \cC\cN(\b0, \bR_{mk} - \bPhi_{mk})$ is the channel estimation error vector.
	
	\begin{remark}
		It is worth mentioning that $ \hatvh_{mk,n} $ and $ \tilde{\bh}_{mk,n} $ are neither independent nor jointly complex Gaussian vectors contrary to conventional estimation theory concerning independent Gaussian noise but they are uncorrelated and each of them has zero mean. The reasoning relies on the fact that the effective distortion noises, e.g., $ \bTheta_{mi,n} \bh_{mi} \delta_{\mathrm{t},n}^{i}$ are not Gaussian since they appear as products between two Gaussian variables. In other words, we have not derived the optimal MMSE estimator but the suboptimal LMMSE which, in general, is accompanied by a little performance loss~\cite{massivemimobook}. In the special case of ideal hardware, the LMMSE estimator becomes the optimal MMSE estimator. 
	\end{remark}
	\begin{remark}
		All the UEs will interfere with each other in the training phase, even those that have been assigned mutually orthogonal pilot sequences. The reason behind this property is that the transmitter and receiver distortions break the orthogonality.
	\end{remark}
	
	%\textbf{Clearly, the channel estimate in~\eqref{EstimatedChannel} depends on time $ n $, which means that any decoders $ \bv_{k,n} $ should be computed in every symbol interval of the data transmission phase but this is computationally prohibitive. For this reason, we obtain the decoders for one symbol interval $ n=\tau+1 $ by means of the channel estimate at time $ \tau $, and then, we apply them during the whole data transmission phase. Specifically, we denote $ \bv_{k,n}=f\left(k,\tau\right) $ for $ k=\{1, \ldots, K\} $, where $ f\left(\cdot\right) $ is a function of the channel estimates, dependent on $ \tau $, since they are obtained during the uplink transmission phase.}
	
	\section{Uplink Data Transmission with HWIs}\label{PerformanceAnalysis}
	The central point of this section is the derivation of the achievable SE for a practical SCF mMIMO setup with HWIs.\footnote{The main purpose of SCF mMIMO is to deliver broadband connectivity to a large number of UEs and provide almost uniformly high data rates. As shown in \cite{Bjoernson2015}, the channel capacity (and the type of lower capacity bound considered in this paper) is closely describing the achievable performance in such situations. The reason is that modern channel codes are nearly capacity-achieving when the block size is above 1 kbyte, so the bit-error-rate (BER) with optimal modulation/coding falls rapidly to zero at the point predicted by the capacity. Based on these insights, we follow the standard practice in the field and utilize SE/capacity as the performance metric. Notably, the BER metric cannot provide further insights due to the equivalence.} Specifically, we consider a fully centralized processing architecture, where the APs, acting basically as relays, forward all their signals (received pilot and data signals) to the CPU. In particular, the CPU handles both the channel estimation and data detection. In the uplink detection phase, let the $ m $th AP receive its signal but delegate its detection to the CPU. For this reason, the $ M $ APs have to send their received signals $ \left\{\by_{m,n}: m=1, \ldots, M\right\} $ to the CPU. In such case and by setting $ W=ML $,\footnote{Given that a CF mMIMO architecture assumes that the number of APs is very large, we consider that the total number of antennas $ W $, i.e., the product of the number of APs and the number of antennas per AP, is also very large. Besides that, the ratio of the number of UEs to the total number of antennas $ W $, denoted by $ \beta=K/W $, is assumed to be constant. Hence, the proposed design is eligible for the application of DE tools used below.} the received signal by the CPU at channel use $ n $ of the uplink data transmission phase can be written in a compact form as
	\begin{align}
		\by_{n}=\sum_{i=1}^{K}\bTheta_{i,n}\bh_{i}\left(s_{i,n}+\delta_{\mathrm{t},n}^{i}\right) +\deltav_{\mathrm{r},n}+\bxi_{n},\label{ULTrans}
	\end{align}
	where $ s_{i,n} \in \mathbb{C} $ is the transmit signal from UE $ i $ with power $ \rho_{i} $, $ \by_{n}=\left [\by_{1,n}^{\T}\cdots \by_{M,n}^{\T} \right ]^{\T} \in \mathbb{C}^{W \times 1}$ and $ \bxi_{n}=\left [\bxi_{1,n}^{\T}\cdots \bxi_{M,n}^{\T} \right ]^{\T} \in \mathbb{C}^{W \times 1}$ are block vectors, while $ \bh_{i,}=\left[ \bh_{i1}^{\T}\cdots \bh_{iM}^{\T} \right ]^{\T} \sim \mathcal{CN}\left (\b0, \bR_{i}\right ) $ is the concatenated channel vector from all APs. Besides, $ \bR_{i}=\diag\left (\bR_{i1}, \ldots,\bR_{iM}\right )\in\mathbb{C}^{W \times W} $ is the block diagonal spatial correlation matrix by assuming that the channel vectors of different APs are independently distributed.
	%\footnote{\textbf{Notably, although~\eqref{ULTrans} is mathematically equivalent to a single-cell mMIMO system with correlated fading, the pilot allocation, and the generation of the correlation matrices are different. Especially, in a single-cell setting, orthogonal pilots are assigned among the UEs, while in a CF mMIMO network, we meet pilot contamination among the UEs served by the same AP since multiple UEs can be assigned with the same pilot.}}. 
	In a similar way, we express the estimated channel and estimated error covariance matrices by $ \bPhi_{i}=\diag\left (\bPhi_{i1}, \ldots,\bPhi_{iM}\right )\in\mathbb{C}^{W \times W} $ and $ \tilde{\bR}_{i}=\bR_{i}- \bPhi_{i} \in\mathbb{C}^{W \times W} $, respectively. Moreover, regarding to the HWIs, we have the PN block diagonal matrix $ \bTheta_{i}=\diag\left (\bTheta_{i1}, \ldots,\bTheta_{iM}\right )\in\mathbb{C}^{W \times W} $. Also, $ \delta_{\mathrm{t},n}^{i} $ is the transmit additive distortion from the $ i $th UE, and $\deltav_{\mathrm{r},n}=\left [{\deltav_{\mathrm{r},n}^{1}}^{\!\!\!\!\T},\cdots, {\deltav_{\mathrm{r},n}^{M}}^{\!
		\T} \right ]^{\T} \in \mathbb{C}^{W\times 1}$ is the receive additive distortion block vector. Note that $ \bxi_{n}\sim \mathcal{CN} \left (\b0,\bF_{\bxi}\right ) $ with $ \bF_{\bxi}=\mathrm{diag}\left( \xi_{1}\Id_{L},\ldots,\xi_{M}\Id_{L} \right) $.\footnote{Henceforth, we assume that all APs are constructed with the same hardware, and thus, are degraded from the same impairments, e.g., $ \xi_{i}=\xi ~\forall i$, to facilitate the analysis. However, the results can be easily extended to the general case where different APs have different hardware parameters. A similar assumption holds for the hardware parameters across the UEs.}
	
	According to the DCC framework~\cite{Bjornson2011}, although all APs receive the signals from all UEs, only a subset of the APs contributes to signal detection. Thus, the network estimates $ s_{k,n} $ by means of~\eqref{ULTrans} as
	\begin{align}
		\hat{s}_{k,n}&\!=\!\sum_{m=1}^{M}\bv_{mk,n}^{\H}\bD_{mk}\by_{m,n}\label{estimatedSignal0}\\
		&\!=\!\underbrace{\bv_{k,n}^{\H}\bD_{k}\bh_{k,n}s_{k,n}}_{Desired~signal}\!+\!\!\!\underbrace{\sum_{i=1,i\ne k}^{K}\!\!\!\bv_{k,n}^{\H}\bD_{k}\bh_{i,n}s_{i,n}}_{Multi-user~interference}\!\nn\\
		&+\!\sum_{i=1}^{K}(\bv_{k,n}^{\H}\bD_{k}\bh_{i,n}\!\!\!\!\!\!\!\underbrace{\delta_{\mathrm{t},n}^{i}}_{{\scriptsize \begin{aligned} Transmit\\~ distortion\end{aligned}}}\!\!\!\!\!\!\!)+\bv_{k,n}^{\H}\bD_{k}(\!\!\!\!\!\underbrace{\deltav_{\mathrm{r},n}}_{{\scriptsize \begin{aligned} Receive\\ ~distortion\end{aligned}}}\!\!\!\!\!\!\!+\!\!\overbrace{\bxi_{n}\!}^{{\scriptsize ATN}}) \label{estimatedSignal} 
	\end{align}
	with $ \bv_{k,n}=\left [\bv_{1,n}^{\T}\cdots \bv_{M,n}^{\T} \right ]^{\T} \in \mathbb{C}^{W \times 1}$ being the collective combining vector and $ \bD_{k} =\diag\left (\bD_{1k}, \ldots,\bD_{Mk}\right )\in\mathbb{C}^{W \times W} $ being a block diagonal matrix while \eqref{estimatedSignal} is obtained from \eqref{estimatedSignal0} in terms of collective vectors. 
	% According to \eqref{estimatedSignal}, each AP preprocesses its signal and computes local estimates of the data, which are then forwarded to the CPU for final decoding.
	
	Aiming at the study of the performance of the realistic aspect of SCF mMIMO to reveal the fundamental impact of non-ideal hardware, we focus on upper and lower bounds since the ergodic capacity is unknown for this system model. Given that the channels change with channel use $ n $, the general procedure for both bounds consists of obtaining one SE for each channel use $ n $ in the transmission phase, and then, taking the average of the SE over channel uses as in~\cite{Pitarokoilis2015}.
	\subsection{Upper bound}
	Following a genie-aided procedure similar to~\cite[Lem. 1]{Bjoernson2014}, we assume that the uplink pilots provide each AP with perfect CSI and that the interference terms can somehow be cancelled. Then, we infer that single-stream Gaussian signaling maximizes the resulting upper-bound on the mutual information since all the additive noise terms (the transmit and receive distortion as well as the ATN) are circularly symmetric complex Gaussian
	distributed and independent of the desired signal.
	\begin{proposition}\label{Proposition:Upperbound}
		An upper bound on the uplink capacity 	in a practical CF mMIMO system is given by
		\begin{align}
			\mathrm{SE}_{k}^{\mathrm{up}}=\frac{1}{\tau_{\mathrm{c}}}\sum_{n=1}^{\tau_{\mathrm{c}}-\tau_{\mathrm{p}}}\EE\left\{\log_{2}\left (1+\gamma^{\mathrm{up}}_{k,n}\left( \bv_{k,n} \right)\right )\right\}\!,\label{Upper1}
		\end{align}
		where the instantaneous SINR $ \gamma^{\mathrm{up}}_{k,n}\left( \bv_{k,n} \right) $ is given by
		%	\sum_{i=1, i\ne k}^{K}\!\!\!\rho_{i}{\bh}_{i,n}{\bh}_{i,n}^{\H}\!+
		\begin{align}
			&\gamma^{\mathrm{up}}_{k,n}\!\left( \bv_{k,n} \right)\!=\!\rho_{k}\bv_{k,n}^{\H}\bD_{k} \tilde{\bSigma}
			\bD_{k}\bv_{k,n}\!\label{Upper3}
		\end{align}
		with $\tilde{\bSigma}^{\dagger}\!\!=\!\!\left(\displaystyle \sum_{i=1, i\ne k}^{K}\!\!\!\!\!\rho_{i} \bD_{k\!}\!\left(\kappa_{\mathrm{t}}^{2}\bh_{i,n}{\bh}_{i,n}^{\H}\!+\!\bkappa_{\mathrm{r}}^{2} \bF_{|\bh_{i}|^{2}}\!+\!\xi_{n}\Id_{W}\right)\! \bD_{k}\!\! \right)\! $, and the receive combining vector being
		\begin{align}
			\bv_{k,n}^{\mathrm{MMSE,up}}\!&=\rho_{k}\Bigg(\displaystyle \sum_{i=1}^{K}\!\rho_{i}\bD_{k}\!\Big(\kappa_{\mathrm{t}}^{2}\bh_{i,n}{\bh}_{i,n}^{\H}\!+\!\bkappa_{\mathrm{r}}^{2} \bF_{|\bh_{i}|^{2}}+\!\xi\Id_{W}\Big)\bD_{k}\! \Bigg)^{\!\!\dagger} \nn\\
			&\times \bD_{k}{\bh}_{k,n},\label{OptimalV}
		\end{align}
		where $ \bkappa_{\mathrm{r}}^{2}\!=\!\Id_{L}\otimes \mathrm{diag}\!\left(\!\kappa_{\mathrm{r}_{1}}^{2},\ldots,\kappa_{\mathrm{r}_{M}}^{2}\! \right)$. Also, we denote $ \bF_{\!|\bh_{i}|^{2}}\!=\!\mathrm{diag}\left( \bF_{|\bh_{1i}|^{2}},\ldots, \bF_{|\bh_{Mi}|^{2}} \right)$, where the elements of the block diagonal matrix are given by 	$ \bF_{|\bh_{mi}|^{2}}=\mathrm{diag}\left( |h_{mi}^{\left(1\right)}|^{2},\ldots,|h_{mi}^{\left(L\right)}|^{2} \right)$ with $ \bh_{mi}=\left[h_{mi}^{\left(1\right)},\ldots,h_{mi}^{\left(L\right)}\right]^{\T} $.
		%	\begin{align}
		%\bv_{k,n}^{\mathrm{LMMSE}}=\frac{\left(\displaystyle\sum_{i=1}^{K}\left (1+\kappa_{\mathrm{t}_\mathrm{UE}}^{2}\right )p_{i}{\bh}_{i,n}{\bh}_{i,n}^{\H}+\bkappa_{\mathrm{r}}^{2}\sum_{i=1}^{K}p_{i} \bF_{|\bh_{i}|^{2}}+\bF_{\bxi_{n}} \right)^{-1}\bD_{k}{\bh}_{k,n}}{\Bigg\|\left(\displaystyle\sum_{i=1}^{K}\left (1+\kappa_{\mathrm{t}_\mathrm{UE}}^{2}\right )p_{i}{\bh}_{i,n}{\bh}_{i,n}^{\H}+\bkappa_{\mathrm{r}}^{2}\sum_{i=1}^{K}p_{i} \bF_{|\bh_{i}|^{2}}+\bF_{\bxi_{n}} \right )^{-1}\bD_{k}{\bh}_{k,n}\Bigg\|_{2}}\label{OptimalV}
		%	\end{align}
	\end{proposition}
	\begin{proof}
		The proof is provided in Appendix~\ref{Proposition2}.
	\end{proof}
	
	Note that the pre-log sum in \eqref{Upper1} defines the fraction of samples per coherence block used for uplink data transmission. Clearly, the SINR-maximizing combiner in~\eqref{OptimalV} is not a scalable solution since its complexity increases with $ K $. 
	
	Following the approach in~\cite{Bjoernson2019a}, we assume that~\eqref{OptimalV} should include only the UEs which are served by the same APs because the interference afflicting UE $ k $ is basically a result of a small subset of the other UEs. The set of these users is defined as~\cite{Bjoernson2019a}
	\begin{align}
		\mathcal{P}_{k}=\left\{i: \bD_{k}\bD_{i}\ne \b0\right\}.
	\end{align}
	
	Hence, the scalable partial MMSE (PMMSE) version of~\eqref{OptimalV} can be written as
	\begin{align}
		\bv_{k,n}^{\mathrm{PMMSE,up}}\!&=\!\rho_{k} \bar{\bSigma}\bD_{k}{\bh}_{k,n},\label{OptimalV1}
	\end{align}
	where $ \bar{\bSigma}^{\dagger}\!\!=\!\!\bigg(\displaystyle \!\bD_{k}\Big(\sum_{i \in \mathcal{P}_{k}}\!\!\kappa_{\mathrm{t}}^{2}\rho_{i}{\bh}_{i,n}{\bh}_{i,n}^{\H}\!+\!\bkappa_{\mathrm{r}}^{2}\sum_{i \in \mathcal{P}_{k}}\!\!\rho_{i} \bF_{|\bh_{i}|^{2}}\!+\!\xi\Id_{W}\!\Big) \!\bD_{k}\!\!\bigg).$ The advantage of $ \bv_{k,n}^{\mathrm{PMMSE,up}} $ is that it does not scale with $ K $. It is worthwhile to mention that in the special case where all APs are serving user $ k $, transmitting to the same set of UEs, we obtain $ |\mathcal{P}_{k}|=\tau_{\mathrm{p}} $ \cite{Bjoernson2019a}.
	
	This decoding vector, based on the uplink SINR with perfect CSI is not indicated since it will not result in representative conclusions since it is obtained by an upper bound.\footnote{The algorithms derived in this paper are generally making use of imperfect CSI obtained by channel estimation. However, herein, by replacing the imperfect CSI with the perfect CSI, we obtain an upper bound that can be utilized to evaluate the performance penalty of having imperfect CSI.} The proposition in the following subsection presents the uplink MSE, which will allow the derivation of the practical optimal decoder by taking into account the inevitable imperfect CSI.
	
	\subsection{Optimal scalable MMSE combiner}\label{scalableMMSE1}
	Herein, we derive the HA-PMMSE combiner by minimizing the MSE.
	\begin{lemma}\label{MSElemma}
		The uplink MSE for UE $ k $ in the case of SCF mMIMO systems with HWI, conditioned on the imperfect channel estimates $ \hat{\bH}_{n}=\left[\hat{\bh}_{1,n} \ldots \hat{\bh}_{K,n}\right] $, is given by
		\begin{align}
			\mathrm{MSE}_{k,n}=\tr\left(\bC_{k,n}\right)\!,\label{MSElemma1}
		\end{align}
		where $ \bC_{k,n} $ is the error covariance matrix given by
		\begin{align}
			\bC_{k,n}&={{\bv}_{k,n}^{\H}\bD_{k}}{}\Big(\displaystyle\sum_{i=1}^{K}\rho_{i}\left (1+\kappa_{\mathrm{t}}^{2}\right )\left(\hat{\bh}_{i,n}\hat{\bh}_{i,n}^{\H}+\tilde{\bR}_{i}\right)\nn\\&+\bkappa_{\mathrm{r}}^{2}\sum_{i=1}^{K}\rho_{i}\left( \bF_{|\hat{\bh}_{i}|^{2}}+\bF_{|\tilde{\bR}_{i}|^{2}}\right)+\xi\Id_{W} \Big)\bD_{k}{\bv}_{k,n}\nn\\
			&-\rho_{k}{\bv}_{k,n}^{\H}\bD_{k}\hat{\bh}_{k,n}-\rho_{k}\hat{\bh}_{k,n}^{\H}\bD_{k}{\bv}_{k,n}+\rho_{k}.\label{errorcovariance}
		\end{align}
	\end{lemma}
	\begin{proof}
		The proof is provided in Appendix~\ref{Lemma3}.	
	\end{proof}
	
	Taking advantage of Lemma~\ref{MSElemma}, we formulate the optimization problem for the minimization of the sum MSE as
	\begin{equation}
		\begin{aligned}
			&\min_{\rho_{i},\bv_{i,n}\forall i}&&\sum_{i=1}^{K}\mathrm{MSE}_{k,n}\\
			%& \text{subject to}&& \sum_{i=1}^{K} p_{i }\le K p,\\
			&&& \rho_{i}\ge 0,\; \forall i,
		\end{aligned}\label{optMSE}
	\end{equation}
	where $ \mathrm{MSE}_{k,n} $ is given by Lemma \ref{MSElemma}. Given that the combiner vector of each UE affects only the MSE of that respective UE, we focus on the minimization of $ \mathrm{MSE}_{k,n} $ \cite{Zarei2017}. Hence, the following proposition provides the optimal scalable combiner by minimizing the individual $ \mathrm{MSE}_{k,n} $ for UE $ k $.\footnote{It is worthwhile to mention that the seminal works in CF mMIMO systems developed power control algorithms with increasing complexity with the number of UEs, being unfeasible in practice. }
	
	%Since the decoder of each UE affects only the corresponding MSE
	\begin{proposition}\label{MMSE_decoder}
		The HA-PMMSE receive combining vector at channel use $ n $, taking into account the transceiver HWIs in SCF mMIMO systems, is written as
		\begin{align}
			&\!\!\!	\bv_{k,n}^{\mathrm{HA-PMMSE}}\!=\!\rho_{k}\Big(\!\!\displaystyle\sum_{i \in \mathcal{P}_{k}}\!\!\rho_{i}\!\left (1\!+\!\kappa_{\mathrm{t}}^{2}\right )\! \nn\\&\times\!\bD_{k}\Big(\hat{\bh}_{i,n}\hat{\bh}_{i,n}^{\H}\!+\!\bkappa_{\mathrm{r}}^{2}\rho_{i} \bF_{|\hat{\bh}_{i}|^{2}}\!+\!\xi\Id_{W} \!\Big) \bD_{k}+\!\tilde{\bDelta}\!\Big)^{\!\dagger}\bD_{k}{\hat{\bh}_{k,n}},\label{OptimalMMSE}
		\end{align}
		where $ \tilde{\bDelta} = \displaystyle\sum_{i \in \mathcal{P}_{k}} \bD_{k} \tilde{\bDelta}_{i} \bD_{k}$ with $ \tilde{\bDelta}_{i} $ being the covariance matrix including the estimation error matrix and HWIs, which is given by
		\begin{align}
			\tilde{\bDelta}_{i}=\rho_{i}\left (1+\kappa_{\mathrm{t}}^{2}\right )\tilde{\bR}_{i}+\bkappa_{\mathrm{r}}^{2}\rho_{i}\bF_{|{\tilde{\bR}}_{i}|^{2}}.
		\end{align}
	\end{proposition}
	\begin{proof}
		The derivation of the HA-PMMSE relies on the minimization of the $ \mathrm{MSE}_{k,n} $. Taking the derivative of the expression in~\eqref{errorcovariance} with respect to $ {\bv}_{k,n} $ and setting it to zero, i.e., $ 	\pdv{	\bC_{k,n}}{\bv_{k,n}} =0$, we obtain
		\begin{align}
			&	\bD_{k}\Big(\displaystyle\sum_{i=1}^{K}\!\rho_{i}\!\left (1\!+\!\kappa_{\mathrm{t}}^{2}\right )\!\!\left(\hat{\bh}_{i,n}\hat{\bh}_{i,n}^{\H}\!+\!\tilde{\bR}_{i}\right)\!
			\nn\\&+\!\bkappa_{\mathrm{r}}^{2}\sum_{i=1}^{K}\!\rho_{i}\!\!\left( \bF_{|\hat{\bh}_{i}|^{2}}\!+\!\bF_{|\tilde{\bR}_{i}|^{2}}\right)\!+\!\xi\Id_{W} \Big)\bD_{k}{\bv}_{k,n}=
			\rho_{k}\bD_{k}\hat{\bh}_{k,n},\label{errorcovariance1}
		\end{align}
		which results in the optimal decoder given by~\eqref{OptimalMMSE} after taking into consideration the scalability as described in~\eqref{OptimalV1}.
		%	The proof is provided in Appendix~\ref{Proposition3}.
	\end{proof}
	
	Substitution of~\eqref{OptimalMMSE} into~\eqref{optMSE} results in the optimization problem depending only on the transmit powers $ \rho_{i} $. After finding the transmit powers, the optimal decoder is given by~\eqref{OptimalMMSE}. However, this is a challenging problem. Since the central point of this work is the scalable implementation of CF mMIMO systems, we have to apply a scalable power control algorithm. Heuristic scalable algorithms for CF mMIMO systems are found in~\cite{Nayebi2017,Buzzi2017a}, but their comparative evaluation is out of the scope of this work. Herein, for the sake of exposition, we resort to the basic equal power allocation which is a scalable solution, i.e., $\rho_{i}=\rho~\forall i $ ~\cite{Bjoernson2019a}. In such a case, we obtain
	
	\begin{align}
		&\bv_{k,n}^{\mathrm{HA-PMMSE}}=\Big(\displaystyle\sum_{i \in \mathcal{P}_{k}}\left (1+\kappa_{\mathrm{t}}^{2}\right )\bD_{k}
		\nn\\&\times\Big(\hat{\bh}_{i,n}\hat{\bh}_{i,n}^{\H}+\bkappa_{\mathrm{r}}^{2}\sum_{i \in \mathcal{P}_{k}} \bF_{|\hat{\bh}_{i}|^{2}}+\frac{\xi}{\rho}\Id_{W}\Big)\bD_{k}+\frac{1}{\rho}\tilde{\bDelta} \Big)^{\dagger}\bD_{k}{\hat{\bh}_{k,n}}{}.\label{OptimalMMSE1}
	\end{align}
	
	For the reason of comparison, we consider the conventional MMSE decoder $ \bv_{k,n}^{\mathrm{HU-MMSE}} $ and scalable MMSE decoder $ \bv_{k,n}^{\mathrm{HU-PMMSE}} $ in CF mMIMO systems studied in~\cite{Bjoernson2019} and \cite{Bjoernson2019a}, respectively. Hence, the benchmark decoders are 
	\begin{align}
		\bv_{k,n}^{\mathrm{HU-MMSE}}&=\left(\displaystyle\sum_{i =1 }^{K}\hat{\bh}_{i,n}\hat{\bh}_{i,n}^{\H}+\bZ_{k}\! \right)^{-1}\hat{\bh}_{k,n},\label{MMSEC}\\
		\bv_{k,n}^{\mathrm{HU-PMMSE}}&\!=\!\left(\displaystyle\sum_{i \in \mathcal{P}_{k}}\!\!\bD_{k}\hat{\bh}_{i,n}\hat{\bh}_{i,n}^{\H}\bD_{k}\!+\!\bar{\bZ}_{k} \!\right)^{\!\!\!\dagger}\!{\bD_{k}\hat{\bh}_{k,n}}{},\label{MMSES}
	\end{align}
	where $ \bZ_{k}\!=\!\bD_{k}\!\left(\displaystyle\sum_{i =1 }^{K}\!\tilde{\bR}_{i}\!+\!\frac{1}{\rho}\Id_{W}\!\!\right) \!\bD_{k} $ and $ 	\bar{\bZ}_{k}\! =\!\bD_{k}\left(\displaystyle\sum_{i \in \mathcal{P}_{k}}\!\tilde{\bR}_{i}\!+\!\frac{1}{\rho}\Id_{W}\!\!\right) \!\bD_{k} $.
	\begin{remark}
		Clearly, the decoders, described by \eqref{MMSEC} and \eqref{MMSES}, not including the HWIs, could be used if the designer is unaware of the hardware impairments. We have provided all these combiners for the sake of a complete comparison. Specifically, comparing \eqref{MMSES} with \eqref{OptimalMMSE1}, the latter expression contains additional terms, which correspond to the HWIs and the channel estimation error due to pilot contamination. Also, \eqref{OptimalMMSE1} includes the ATN. Note that $ \eqref{MMSEC} $ refers to the most basic expression that accounts for neither HWI nor scalability design.
	\end{remark}
	
	\subsection{Lower bound}
	Since the derivation of lower bounds is in general of higher importance than upper bounds such as the bound described by Proposition~\ref{Proposition:Upperbound}, herein, we focus on the derivation of an accurate lower bound by following the typical analysis in mMIMO theory. In particular, according to~\cite{Medard2000}, we rewrite the received signal in terms of the average effective channel $ \EE\left\{\bv_{k,n}^{\H}\bD_{k}\bTheta_{k,n}{\bh}_{k}\right\}$. Hence, we have
	\begin{align}
		\hat{s}_{k,n}\!&=\!\EE\left\{\bv_{k,n}^{\H}\bD_{k}\bTheta_{k,n}{\bh}_{k}\right\}s_{k}\!+\! \left(\right.\!\!\bv_{k,n}^{\H}\bD_{k}\bTheta_{k,n}{\bh}_{k}s_{k}\nn\\&-\EE\left\{\bv_{k,n}^{\H}\bD_{k}\bTheta_{k,n}{\bh}_{k}\right\}s_{k}\!\!\left.\right)+\!\!\sum_{i\ne k}^{K}\!\bv_{k,n}^{\H}\bD_{k}\bh_{i,n}s_{i}\!\nn\\
		&+\!\sum_{i=1}^{K}\!\bv_{k,n}^{\H}\bD_{k}\bh_{i,n}\delta_{\mathrm{t},n}^{i}\!+\!\bv_{k,n}^{\H}\bD_{k}\!\left(\deltav_{\mathrm{r},n}\!+\!\bxi_{n}\right). \label{estimatedSignal1} 
	\end{align}
	
	%We result in a standard lower bound on the uplink average SE provided by the following proposition. 

	\begin{proposition}\label{LowerBound1}
		The uplink average SE for user $ k$ is lower bounded by
		\begin{align}
			\mathrm{SE}_{k}^{\mathrm{lo}}	=\frac{1}{\tau_{\mathrm{c}}}\sum_{n=1}^{\tau_{\mathrm{c}}-\tau_{\mathrm{p}}}\log_{2}\left ( 1+\tilde{\gamma}^{\mathrm{lo}}_{k,n}\right)\!,\label{LowerBound}
		\end{align}
		where $ \tilde{\gamma}^{\mathrm{lo}}_{k,n}=\frac{S_{k,n}}{I_{k,n}}$ with $ S_{k,n} $ and $ I_{k,n} $ being the desired signal power and the interference
		plus noise power, respectively, which are defined as
		\begin{align}
			S_{k,n}&=\rho_{k}|\EE\left\{\bv_{k,n}^{\H}\bD_{k}{\bh}_{k,n}\right\}\!|^{2}\label{sig11}\\
			I_{k,n}&=\rho_{k}\mathrm{Var}\left\{\bv_{k,n}^{\H}\bD_{k}{\bh}_{k,n}\right\}+\sum_{i\ne k}^{K}\rho_{i}\EE\left\{|\bv_{k,n}^{\H}\bD_{k}\bh_{i,n}|^{2}\right\} \nn\\&+\sigma_{\mathrm{t}}^{2}+\sigma_{\mathrm{r}}^{2}\!+\!\xi_{n}\EE\left\{\|\bv_{k,n}^{\H}\bD_{k}\|^{2}\right\}\!,\label{int1}
		\end{align}
		where $\sigma_{\mathrm{t}}^{2}= \sum_{i=1}^{K}\kappa_{\mathrm{t}}^{2}\rho_{i}\EE\left\{|\bv_{k,n}^{\H}\bD_{k}\bh_{i,n}|^{2} \right\}$ and $\sigma_{\mathrm{r}}^{2}\!=\!\EE\left\{\bv_{k,n}^{\H}\bD_{k}\bkappa_{\mathrm{r}}^{2}\!\left(\sum_{i=1}^{K}\rho_{i} \bF_{|\bh_{i}|^{2}}\right)\bD_{k} \bv_{k,n}\right\} $ are the variances of the components corresponding to the additive HWIs at the transmit (UE $ k $) and receive side (output of the MMSE decoder), respectively.
	\end{proposition}
	\begin{proof}
		The proof is provided in Appendix~\ref{Proposition4}.
	\end{proof}

	\section{Deterministic Equivalent Analysis}\label{Deterministic}
	The theory of DEs , which is employed in this work to derive the SE, concerns derivations in the asymptotic limit $K,W \to \infty$ while their ratio ${K}/{W}=\beta$ is fixed and the number of antennas per AP $L $ is finite, i.e., $ M \to \infty $. Also, we assume similar assumptions to \cite[Assump. A1-A3]{Hoydis2013} concerning the covariance matrices under study. Given that CF mMIMO systems consider a large number of APs, $ W $ is quite large and validates us to employ the DE analysis to obtain the DE uplink achievable rate. The main advantage of this analysis is that the corresponding results are tight approximations, even for moderate system dimensions, e.g., an $ 8 \times 8 $ matrix~\cite{Couillet2011}. Also, the DE expressions make lengthy Monte-Carlo simulations unnecessary. Hence, the extracted conclusions are of significant importance. 
	
	The DE of the SINR $\tilde{\gamma}^{\mathrm{lo}}_{k,n}$, provided by Proposition~\ref{LowerBound1}, obeys to $\tilde{\gamma}^{\mathrm{lo}}_{k,n}-\bar{\gamma}^{\mathrm{lo}}_{k,n}\xrightarrow[M \rightarrow \infty]{\mbox{a.s.}}0$, while the DE of the SE of user $k$, relied on the dominated convergence and the continuous mapping theorem~\cite{Couillet2011}, is given by
	\begin{align}
		\mathrm{SE}_{k}^{\mathrm{lo}}-\mathrm{\overline{SE}}_{k}^{\mathrm{lo}}\xrightarrow[ \rightarrow \infty]{\mbox{a.s.}}0,\label{DeterministicSumrate}
	\end{align}
	where $\mathrm{\overline{SE}}_{k}^{\mathrm{lo}}= \frac{1}{T_{\mathrm{c}}}\sum_{n=1}^{T_{\mathrm{c}}-\tau}\log_{2}(1 + \bar{\gamma}^{\mathrm{lo}}_{k,n}) $.
	Taking into account Proposition~\ref{MMSE_decoder}, we have $\bv_{k,n}= \bv_{k,n}^{\mathrm{HA-PMMSE}} $ while we rewrite the HA-PMMSE decoder at channel use $ n $ as
	\begin{align}
		\bv_{k,n}^{\mathrm{HA-PMMSE}}=\bSigma\bD_{k}\hat{\bh}_{k,n},\label{OptimalMMSE2}
	\end{align}
	where $\bSigma^{\dagger}= \displaystyle\sum_{i \in \mathcal{P}_{k}}\rho_{i}\left (1+\kappa_{\mathrm{t}}^{2}\right )\bD_{k}\Big(\hat{\bh}_{i,n}\hat{\bh}_{i,n}^{\H}+\bkappa_{\mathrm{r}}^{2}\rho_{i} \bF_{|\hat{\bh}_{i}|^{2}}+\al W \xi\Id_{W}\Big)\bD_{k}+\tilde{\bDelta} $ with $\al$ being a regularization scaled by $W$ to make expressions converge to a constant as $W$, $K\to \infty$. A similar regularization takes place during the simulations for the benchmark decoders given by \eqref{MMSEC} and \eqref{MMSES}, respectively.
	\begin{Theorem}\label{theorem:ULDEMMSE}
		The uplink DE of the SINR of user $k$ at channel use $n$ with HA-PMMSE decoding in SCF mMIMO systems, accounting for imperfect CSI and transceiver HWIs, is given by
		\begin{align}
			\!\!\!	\bar{\gamma}_{k,n}^{\mathrm{lo}}\!=\!	\frac{\tilde{ \delta}_{k}^{2}}{\displaystyle 		\kappa_{\mathrm{t}}^{2}\tilde{\delta}^{2}_{k}\!+\!\tilde{\zeta}_{ki}\!+\!\left(1\!+\!\kappa_{\mathrm{t}}^{2}\right)\!\!\sum_{i\ne k}^{K}\!\frac{\rho_{i}}{\rho_{k}}\!\tilde{\mu}_{ki}		\!+\!
				\eta^{'}_{k}\sum_{i=1}^{K}\frac{\rho_{i}}{\rho_{k}} e_{i}\!+\!\frac{\tilde{e}_{k}^{'}}{\rho_{k}}}\label{theo1}
		\end{align}
		with $ \delta_{k}=\frac{\left (1+\kappa_{\mathrm{t}_\mathrm{UE}}^{2}\right )}{W}\tr\bD_{k}\bPhi_{k}\bD_{k}\bT $, $ \tilde{\delta}_{k} =\frac{1}{W}\tr \bD_{k}\bPhi_{k}\bD_{k}\bT$, $ 	\tilde{\zeta}_{ki}=\frac{1}{W^{2}}\tr\left(\bD_{k}\left( \bR_{k} - \bPhi_{k}\right)\bD_{k}\bT^{'}_{k}\right) $,	 $ e_{i}=\frac{1 }{W}\tr\left( \bD_{k}^{2}\bkappa_{\mathrm{r}}^{2}\bR_{i} \right) $, 
		$ \zeta_{ki}=\frac{1}{W^{2}}\tr\left(\bD_{k}\bR_{i}\bD_{k}\bT^{'}_{k}\right) $, 	$ \nu_{ki}=\frac{1}{W}\tr\left(\bD_{k}\bPhi_{i}\bD_{k}\bT\right)$, $ 	\mu_{ki}=\frac{1}{W^{2}}\tr\left(\bD_{k}\bPhi_{i}\bD_{k}\bT^{'}_{k}\right) $, $\tilde{\mu}_{ki}=\zeta_{ki}\!+\!\frac{|\nu_{ki}|^{2}\mu_{ki}}{\left(1\!+\!\delta_{i}\right)^{2}}\!-\!2\mathrm{Re}\left\{\!\!
		\frac{ \nu_{ki}^{*}\mu_{ki}}{\left(1\!+\!\delta_{i}\right)}\right\}$, $ \eta^{'}_{k}=\frac{1}{W}\tr\bT^{'}_{k}$, and $ \tilde{e}_{k}^{'}=\frac{1 }{W^{2}}\tr\left( \bT^{''}_{k}\right) $,
		where
		\begin{itemize}
			\renewcommand{\labelitemi}{$\ast$}
			\item $\bT\!=\!\!\Bigg(\!\!\displaystyle\sum_{i \in \mathcal{P}_{k}} \bD_{k}\!\!\left(\!\!\frac{1+\kappa_{\mathrm{t}_\mathrm{UE}}^{2}}{W\!\left(1\!+\!\delta_{i}\right)}\bPhi_{i}\!+\!\frac{\bkappa_{\mathrm{r}}^{2}}{W}\Id_{W}\!\circ\!\bPhi_{i}\!+\!\al \xi\Id_{W}\!\!\right)\!\!\bD_{k}+\frac{1}{W}\tilde{\bDelta}\Bigg)^{\dagger}$,
			\item $ \bT^{'}_{k}=\bT \bD_{k}\bPhi_{k}\bD_{k}\bT+\sum_{i=1}^{K}\frac{\delta^{'}_{i}\bT\bD_{k}\bPhi_{i}\bD_{k}\bT}{W\left(1+\delta_{i}\right)^{2}}$,
			\item $ \bT^{''}_{k}=\bT\bD_{k}\bT +\sum_{i=1}^{K}\frac{\delta_{i}\bT\bD_{k}\bPhi_{i}\bD_{k}\bT}{W\left(1+\delta_{i}\right)^{2}}$,
			\item $ \deltav^{'} =\left(\Id_{K}-\bF\right)^{-1}\bf$ with $ 	\left[\bF\right]_{k,i}=\frac{1}{W^{2}\left(1+ \delta_{i}\right)}\tr\left(\bD_{k}\bPhi_{k}\bD_{k}\bT\bD_{k}\bPhi_{i}\bD_{k}\bT\right) $
			%			 and $ \left[\bff\right]_{k}=\frac{1}{W}\tr\left(\bD_{k}\bPhi_{k}\bD_{k}\bT\bD_{k}\bPhi_{k}\bD_{k}\bT\right) $.		
		\end{itemize} 
		while all the covariance matrices are assumed with uniformly bounded spectral norms with
		respect to $ W $.
	\end{Theorem} 
	\proof The proof is provided in Appendix~\ref{theorem1}.\endproof
	
	Theorem \ref{theorem:ULDEMMSE} includes clearly the impact of HWIs in the achievable SE as clearly can be seen. It is a general expression assuming no power allocation. For the sake of exposition, the numerical results below will assume $\rho_{i}=\rho~\forall i $ as described in Section \ref{scalableMMSE1}. 	Furthermore, we would like to mention that the analytical expression for the spectral efficiency has a common complexity compared to works on DEs. Note that Theorem 1 achieves to provide the closed-form SE with MMSE decoding, which otherwise can be studied only by simulation. Moreover, another main advantage of \eqref{theo1} is its deterministic nature, which means that no Monte-Carlo simulations are necessary. Hence, we can study the impact of the hardware impairments very fast and easily. 
	%\begin{remark}
	%The term $ \sigma_{\mathrm{PN}}$, capturing the effect of PN variation between the transmission and training phases, is given by \cite{Krishnan2015}
	%\begin{align}
	%\sigma_{\mathrm{PN}}&=\frac{1}{W} \sum_{i=1}^{W} e^{i \left(\theta_{mk,n}^{(i)}-\theta_{mk,\tau}^{(i)}\right)}\\
	%&\xrightarrow[ W \rightarrow \infty]{} \left\{
	%\begin{array}{ll}
	%e^{-\frac{\sigma_{\theta}^{2}}{2}\Delta_{n}} &\mathrm{SLO~setup }\\
	%e^{-j \left(\theta_{mk,n}-\theta_{mk,\tau}\right)} &\mathrm{CLO~setup. }
	% \end{array} 
	%\right.
	%\end{align}
	%In the last step, the law of large numbers has been applied. The expression of $ \bar{\sigma}_{\mathrm{PN}} $ is identical for both scenarios and equal to $ e^{-{\sigma_{\theta}^{2}}\Delta_{n}} $. However, in the case of $ \tilde{\sigma}_{\mathrm{PN}} $, we obtain
	%\begin{align}
	% \tilde{\sigma}_{\mathrm{PN}} =\left\{
	% \begin{array}{ll}
	% e^{-{\sigma_{\theta}^{2}}\Delta_{n}} &\mathrm{SLO~setup }\\
	% 1&\mathrm{CLO~setup,}
	% \end{array} 
	% \right.\label{sloclo}
	%\end{align}
	%where it is shown that the variances decreases from $ 1 $ to $ e^{-{\sigma_{\theta}^{2}}\Delta_{n}} $ when we have CLO and SLO, respectively. In particular, in the case of SLO, $ \tilde{\sigma}_{\mathrm{PN}} $ hardens to a value depending on $ {\sigma_{\theta}^{2}} $ and $ \Delta_{n} $ while there is no averaging in the case of CLO. Hence, the deterministic SINR 	$ \bar{\gamma}_{k,n}^{\mathrm{lo}} $ is lower in the case of CLO and higher when we employ SLOs.
	%\end{remark}
	\section{Numerical Results}\label{results}
	This section aims at demonstrating representative numerical examples quantifying the performance of SCF mMIMO networks with HWIs. As a metric for study and comparison, we consider the achievable SEs per UE, given by Theorem~\ref{theorem:ULDEMMSE} and Proposition~\ref{Proposition:Upperbound}. In parallel, Monte-Carlo simulations, conducted for $ 10^{3} $ independent channel realizations, verify the proposed analytical results and show their tightness. Specifically, although the analytical results rely on the assumption that $W,K \to \infty$, they coincide with the simulations even for finite values of $W$ and $K$ being of practical interest.\footnote{This observation is already known in the literature concerning the DEs and supports their usefulness~\cite{Couillet2011,Hoydis2013,Papazafeiropoulos2015a}.} The parameter values, used in the simulation setup below, correspond to a practical scenario to allow the extraction of meaningful conclusions.

	%\textbf{The figures illustrate the proposed analytical expressions of the average SE per UE along with the corresponding simulated results. The “solid” lines with certain patterns depict the analytical results with specific, the
	%“dot” lines, in most of the figures, correspond to the “ideal” expressions with no transceiver impairments, and the “stars” represent the simulation results.}
	\subsection{Simulation Setup}
	We consider $ M=200 $ APs with $ N=3 $ antennas and $ K=40 $ UEs in an $ 2 \times 2~\mathrm{km}^{2}$ area. Also, we consider the 3GPP Urban Microcell model in~\cite[Table B.1.2.1-1]{3GPP2017} as a proper benchmark for CF mMIMO systems.\footnote{In~\cite[Remark 4]{Bjoernson2019}, it is explained that this propagation model corresponds better to the architecture design of CF mMIMO systems than the established model for CF systems suggested initially in~\cite{Ngo2017}. The main reason is that~\cite{Ngo2017} uses the COST-Hata model, which is suitable for macro-cells with APs being at least
		$ 1~\mathrm{km} $ far from the UEs and at least $ 30~\mathrm{m} $ above the ground. Obviously, these characteristics do not match the CF setting where the APs are very close to the UEs, and possibly, at a lower height. Another important reason is that the model in~\cite{Ngo2017} does not account for shadowing when the UE is closer than $ 50~\mathrm{m} $ from an AP while CF mMIMO systems are more likely suggested for shorter distances.} In particular, according to this model assuming a $ 2~\mathrm{GHz}$ carrier frequency, the large-scale fading coefficient is given by
	\begin{align}
		\beta_{mk}[dB]=-30.5-36.7\log_{10}\left(\frac{d_{mk}}{1~\mathrm{m}}\right)+F_{mk},
	\end{align}
	where $ d_{mk} $ is the distance between AP $ m $ and UE $k $, and $ F_{mk}\sim\mathcal{CN}\left(0,4^{2}\right) $ is the shadow fading. Note that shadowing terms between different UEs are correlated as $\EE\{F_{mk}F_{ij}\}=4^{2}2^{-\delta_{kj}/9} $ when $ m=i $, while if $ m\ne i $, they are uncorrelated. The parameter $ \delta_{kj} $ denotes the distance between UEs $ k $ and $ i $. In addition, we assume that the coherence time and bandwidth are $T_{\mathrm{c}}=2~\mathrm{ms}$ and $B_{\mathrm{c}}=100~\mathrm{kHz}$, respectively, i.e., the coherence block consists of $ 200 $ channel uses with $ \tau_{\mathrm{p}}=20 $ and $ \tau_{\mathrm{u}} =180$. Moreover, all UEs transmit with the same power in both uplink training and transmission phases given by $ \rho=\rho_{\rp}=100~\mathrm{mW} $ while the thermal noise
	variance is $ \sigma^{2}=-174 \mathrm{dBm/Hz}$.
	
	Regarding the HWIs, we assume similar to~\cite{Bjornson2015,Papazafeiropoulos2017,Papazafeiropoulos2019} that the variance of PN is $ \sigma_{i}^{2} =1.58 \cdot 10^{-4}$ by setting $ f_{\mathrm{c}} =2~ \mathrm{GHz}$, $ T_{\mathrm{s}} = 10^{-7} \mathrm{s} $, and $ c_{i}=10^{-17} $ for $ i=\phi,\varphi $ in \eqref{PN1}. Also, we consider an ADC quantizing the received signal to a $b$-bit resolution. In such case, we have $\kappa_{\mathrm{r}}=2^{-b}/\sqrt{1-2^{-2b}}$. For simplicity, we use the same numbers for the transmitter distortion $ \kappa_{\mathrm{t}} $. Note that the trend in 5G networks is to employ low-precision ADCs \cite{Hu2019}. For example, if $ b=2,3,4 $, then $\kappa_{\mathrm{t}}=0.258,~0.126,0.062$. Moreover, we assume that $ \xi=1.6\sigma^{2} $ by considering a low noise amplifier with $\mathcal{F}$ being the noise amplification factor. If $\mathcal{F}=2$ dB and $b=3$ bits, it results in $\xi=\frac{F \sigma^{2}}{1-2^{-2b}}=1.6\sigma^{2}$. Note that the assignment of pilots to the UEs follows the three-step access procedure in \cite{Bjoernson2019a}. The following figures focus on the impact of the PN and additive HWIs in SCF mMIMO systems while the impact concerning the ATN is not shown because it is negligible and due to limited space.

	The SE comparison between the MMSE and PMMSE receivers for different values of $ \rho $, considering the impact of HWIs, is shown in Fig.~\ref{Fig1}. In particular, the receivers, describing the HWIs, are denoted as HA-MMSE and HA-PMMSE, respectively. Note that we include both scalable and conventional versions to show the loss in each case. Also, we provide the HU receivers. Clearly, the proposed HA receivers perform better than the standard MMSE and PMMSE receivers, which are hardware unaware. Hence, the increase of the achievable SE for this simulation setup of HA-MMSE against HU-MMSE is $ 17\% $. In particular, at low SNR due to the low quality of CSI, all decoders perform similarly. As the transmit SNR increases the gap between HA and HU decoders increases because the quality of CSI is improved and the inter-user interference is reduced. As expected, all the rates saturate at high SNR due to the pilot contamination and the power-dependent additive HWIs (see \eqref{eta_tU}, \eqref{eta_rU}). It is shown that the performance loss between scalable and full MMSE is not significant, and the PMMSE version should be preferable given its advantage. Especially at $ 20 $ dB, the loss for the HA-MMSE receiver is $ 8 \%$. In addition, we depict the upper bound on the capacity for both scalable and unscalable MMSE decoders. In the same figure and for reference, we have provided the cases of MRC with/without HWIs, and ideal MMSE and PMMSE when perfect hardware is assumed. The latter case presents the best performance but does not take into account the inevitable HWIs.
	\begin{figure}[!h]
		\begin{center}
			\includegraphics[width=0.95\linewidth]{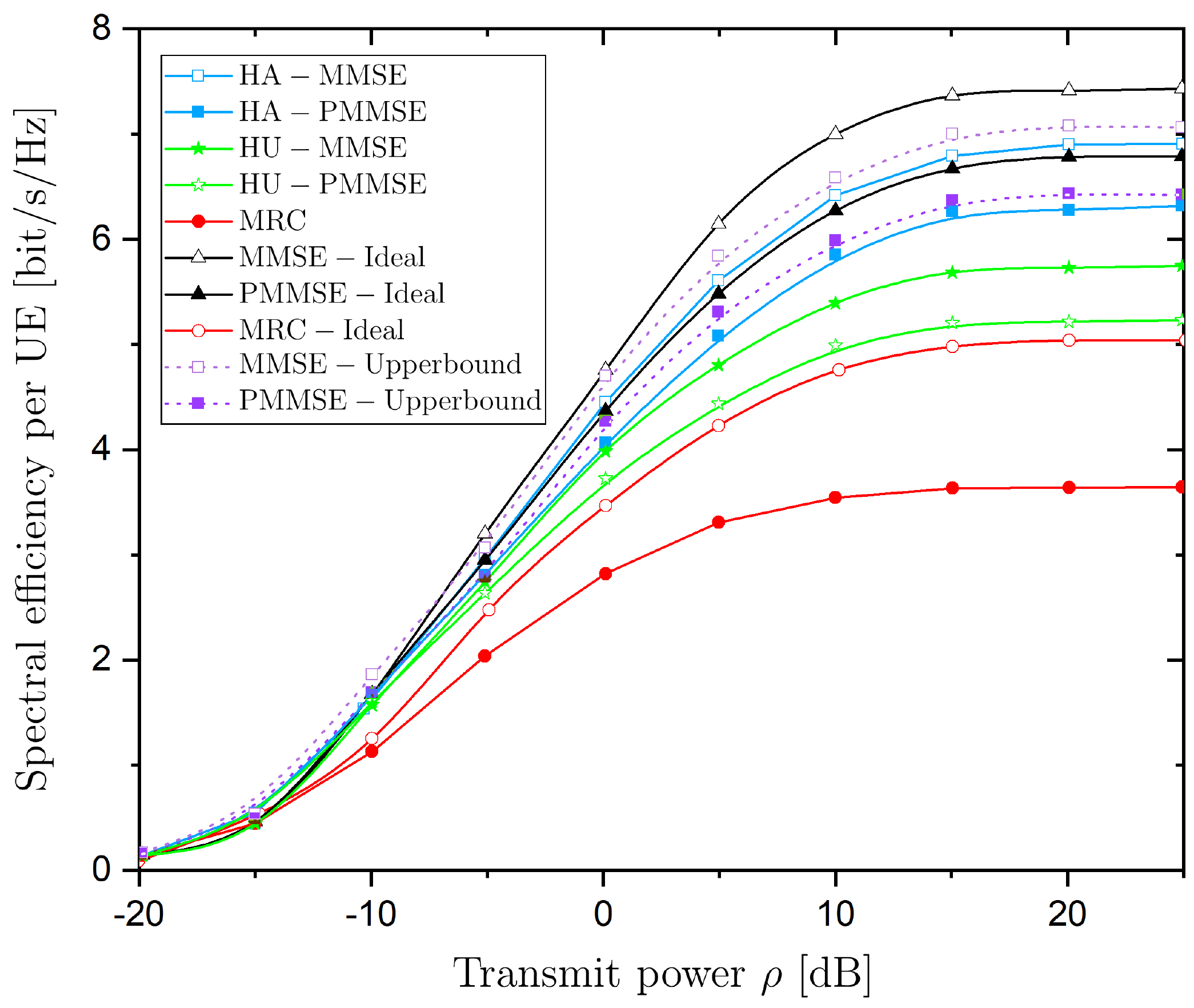}\vskip -5mm
			\caption{\footnotesize{Uplink achievable SE per UE versus the power $\rho$ for SCF mMIMO systems with HA and HU MMSE and PMMSE decoding ($\sigma_{i}^{2} =1.58 \cdot 10^{-4}$ for $ i=\phi,\varphi $, $ \kappa_{\mathrm{t}}=\kappa_{\mathrm{r}}=0.126 $, and $ \xi=1.6\sigma^{2} $).}}
			\label{Fig1}
		\end{center}
	\end{figure}
	\begin{figure}[!h]
		\begin{center}
			\includegraphics[width=0.95\linewidth]{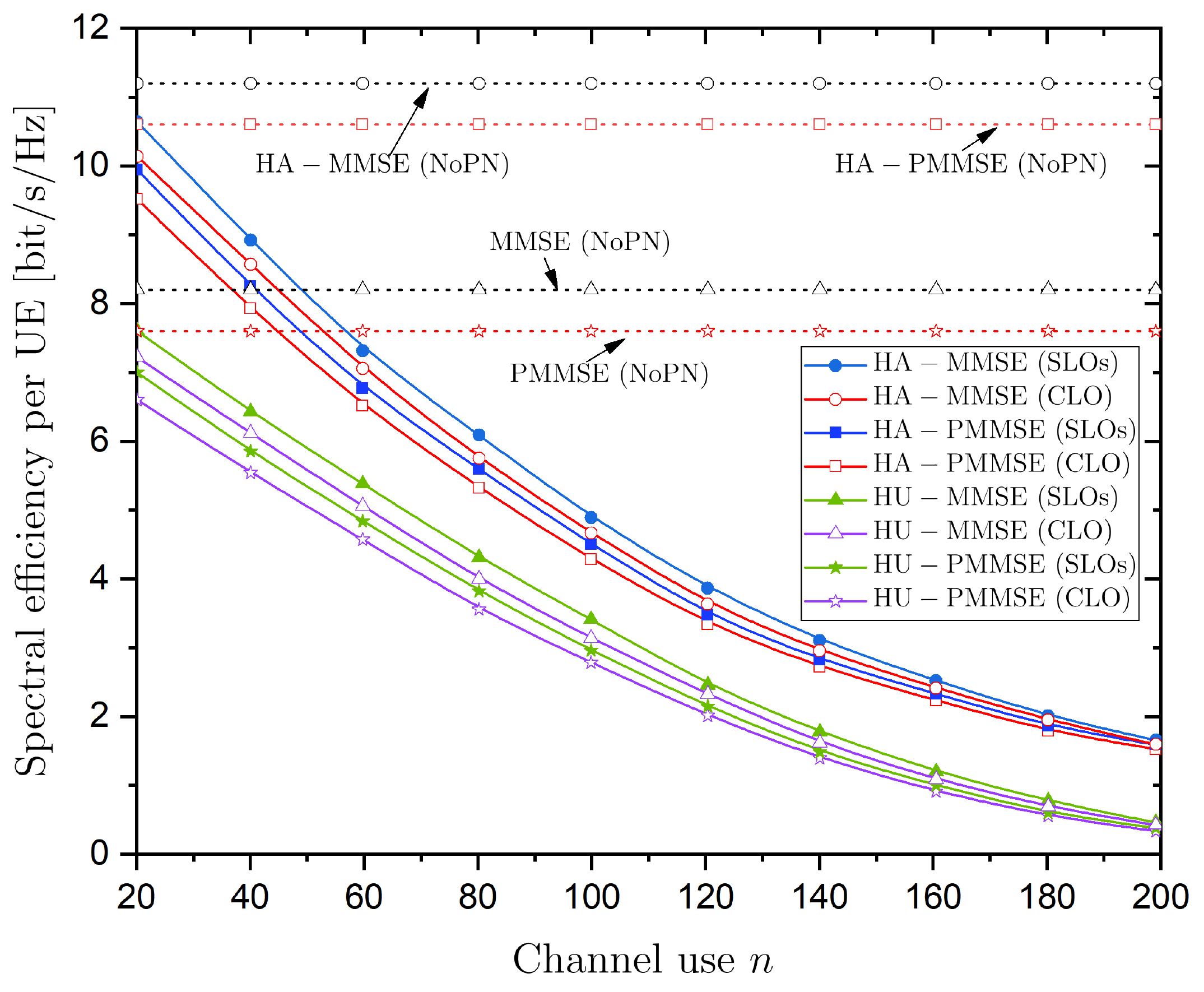}\vskip -5mm
			\caption{\footnotesize{Uplink achievable SE per UE versus the channel use $n$ of the data transmission phase for SCF mMIMO systems with HA and HU MMSE and PMMSE decoding ($\sigma_{i}^{2} =1.58 \cdot 10^{-4}$ for $ i=\phi,\varphi $, $ \kappa_{\mathrm{t}}=\kappa_{\mathrm{r}}=0 $, and $ \xi=1.6\sigma^{2} $).}}
			\label{Fig2}
		\end{center}
	\end{figure}
	\begin{figure}[!h]
		\begin{center}
			\includegraphics[width=0.95\linewidth]{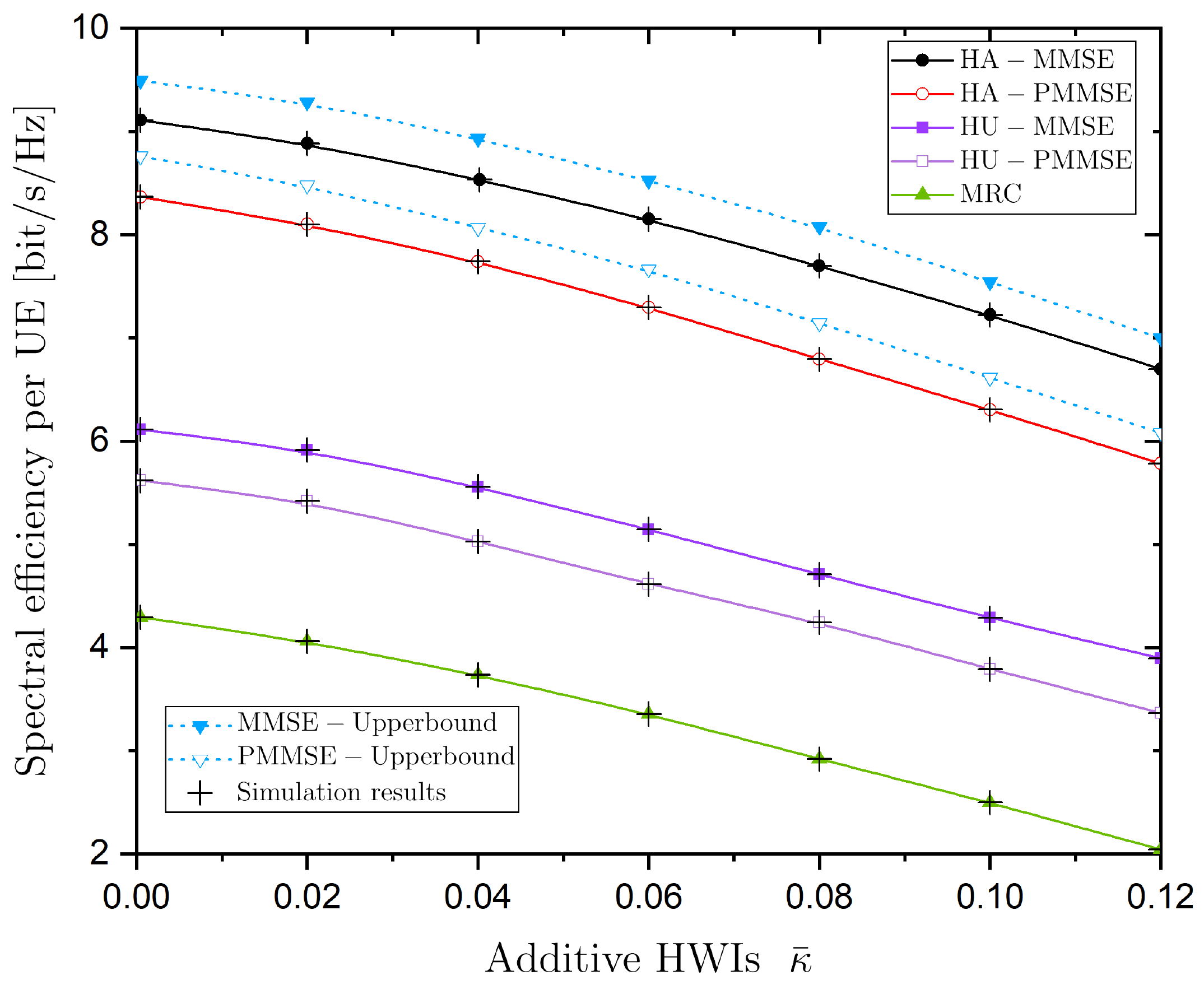}\vskip -5mm
			\caption{\footnotesize{Uplink achievable SE per UE versus $ \bar{\kappa} $ for SCF mMIMO systems with HA and HU MMSE and PMMSE decoding ($\sigma_{i}^{2} =0$ for $ i=\phi,\varphi $, $ \kappa_{\mathrm{t}}=\kappa_{\mathrm{r}}=0.126 $, and $ \xi=1.6\sigma^{2} $).}}
			\label{Fig3}
		\end{center}
	\end{figure}
	\begin{figure}[!h]
		\begin{center}
			\includegraphics[width=0.95\linewidth]{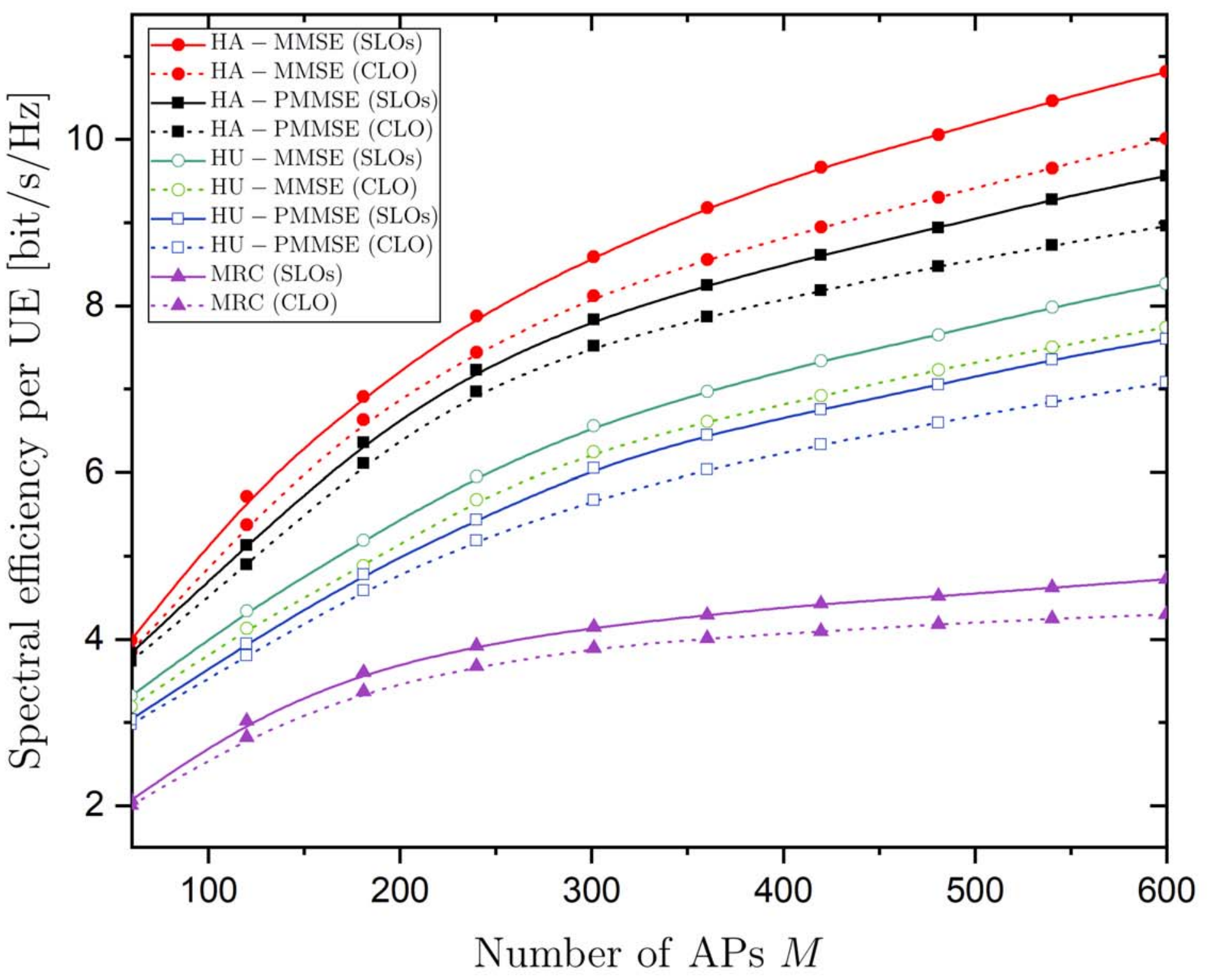}\vskip -5mm
			\caption{\footnotesize{Uplink achievable SE per UE versus the number of APs $M$ for SCF mMIMO systems with HA and HU MMSE and PMMSE decoding ($\sigma_{i}^{2} =1.58 \cdot 10^{-4}$ for $ i=\phi,\varphi $, $ \kappa_{\mathrm{t}}=\kappa_{\mathrm{r}}=0 $, and $ \xi=1.6\sigma^{2} $).}}
			\label{Fig4}
		\end{center}
	\end{figure}
	\begin{figure}[!h]
		\begin{center}
			\includegraphics[width=0.95\linewidth]{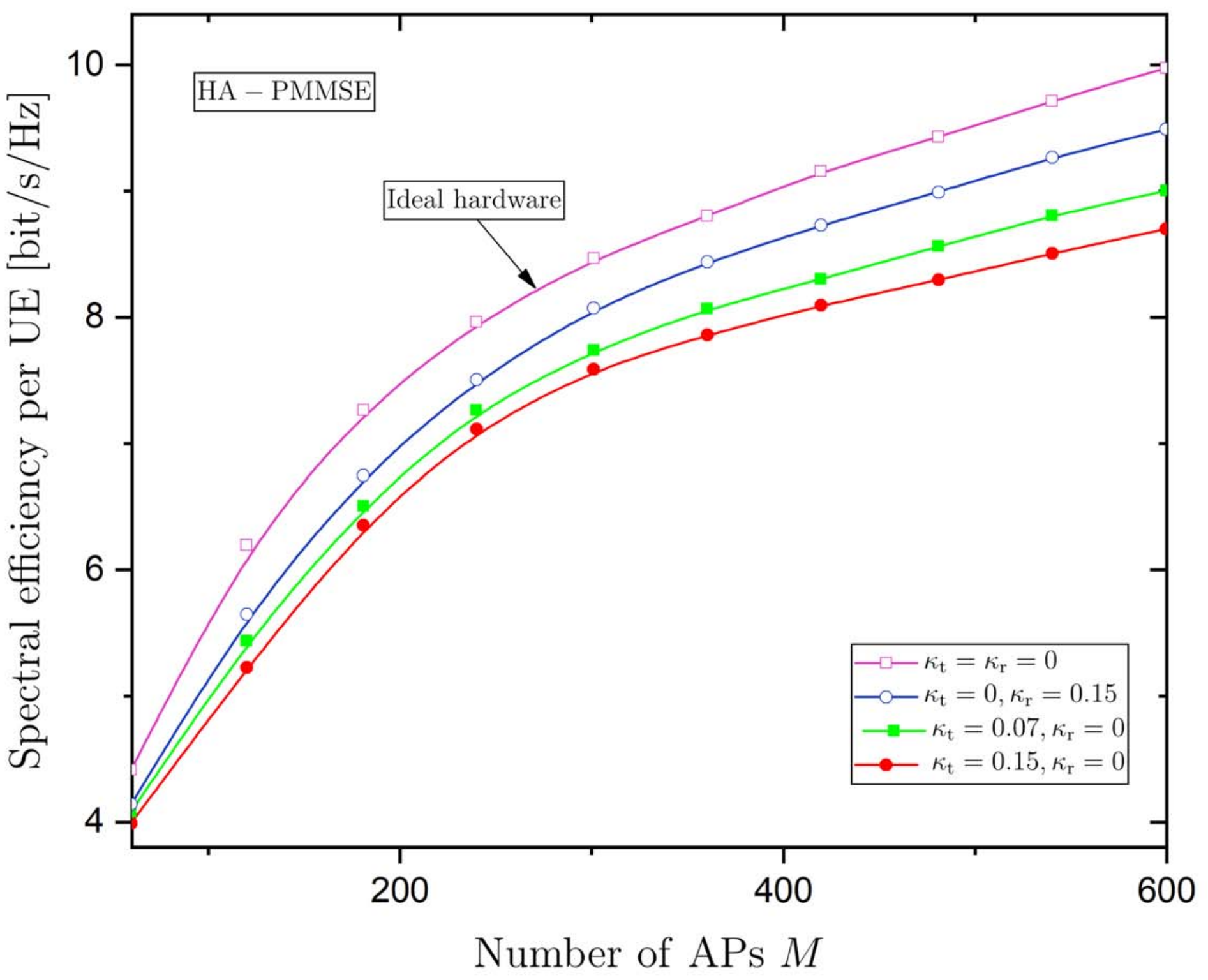}\vskip -5mm
			\caption{\footnotesize{Uplink achievable SE per UE versus the number of APs $M$ for SCF mMIMO systems with HA and HU MMSE and PMMSE decoding ($\sigma_{i}^{2} =1.58 \cdot 10^{-4}$ for $ i=\phi,\varphi $, $ \kappa_{\mathrm{t}}=\kappa_{\mathrm{r}}=0.126 $, and $ \xi=1.6\sigma^{2} $).}}
			\label{Fig5}
		\end{center}
	\end{figure}

	%Fig.~\ref{Fig3} depicts the impact of PN on both scalable and conventional CF mMIMO systems by showing the average SE per user versus the phase drift variance $\sigma_{\phi}^{2} $ of the LOs belonging to the APs while the RAHWIs are assumed negligible. Specifically, as $ \sigma_{\phi}^{2} $ increases, $ \mathrm{\overline{SE}}_{k}^{\mathrm{lo}} $ decreases in all cases. When $ \sigma_{\phi}^{2}=0 $, we observe the rate in the case of perfect LOs (no PN). It is shown that SLOs and CLO designs behave differently, i.e., the rate with SLOs appears superior against the CLO setting since the distortions from different LOs average out in the large system limit described by \eqref{sloclo} \cite{Papazafeiropoulos2016}. However, when the PN variance approached $ x? $, the degradation is so heavy that the performance with scalable and conventional CF mMIMO systems becomes almost identical. In other words, a small phase drift affects little the rate while a large phase drift results in almost zero rates for both SLOs and CLO settings. Hence, only intermediate values of $ \sigma_{\phi}^{2} $ reveal the outperformance of the SLOs setting.
	
	Keeping the additive HWIs equal to zero, Fig.~\ref{Fig2} illustrates the variation of the achievable SE per user versus the channel use of the data transmission phase.
	%and the length of the data transmission interval $\tau_{\mathrm{d}}$, respectively. 
	In particular, in Fig.~\ref{Fig2}, the achievable SE decreases as the number of channel uses increases since the aggregate detrimental contribution from PN becomes higher. We also depict the SEs of MMSE and PMMSE receivers and their HA versions with no PN. These are constant with respect to $ n $ since they do not depend on the PN being the only source of channel aging in this work. 
	%At the beginning of the figure, i.e., when $\tau_{\mathrm{d}}=10$ both the MMSE and HA-MMSE as well as PMMSE and HA-PMMSE couples of lines coincide because we have no PN contribution.
	Next, in both HA-MMSE and HA-PMMSE receivers, the SLOs configuration outperforms the CLO design 
	since the distortions from different LOs average out in the large system limit as described by \cite{Bjornson2015,Papazafeiropoulos2016}. 
	% for the same reason explained in the previous paragraph. 
	Moreover, the SEs for both HA-MMSE and HA-PMMSE receivers diminish with increasing $ n $ since the phase drift becomes higher and the corresponding impact of PN becomes quite destructive. 
	% Shedding further light into the impact of PN on the SE of each receiver, we observe trade-off between $ \bar{R}_{k} $ and $\tau_{\mathrm{d}}$ for both HA-MMSE and HA-PMMSE decoders. na valw ti leei sto fig.
	At the end of the data transmission, where $ n=200 $, the degradation is so heavy that the performance with scalable and conventional CF mMIMO systems becomes almost identical for both cases of HA and HU decoders.

	Fig.~\ref{Fig3} provides the performance of the achievable SE versus $ \bar{\kappa} $, where $ \kappa_{\mathrm{t}}=\bar{\kappa} $ and $ \kappa{\mathrm{r}}=\bar{\kappa} +0.03 $, while the effect of PN has not been considered to focus on the impact of the additive HWIs. The impact of PN is studied in other figures. The figure starts from the case of no additive HWIs at the BS when $ \bar{\kappa}=0 $ and ends with severe additive HWIs. Obviously, the higher the additive HWIs, the higher the degradation of the system performance becomes. Furthermore, it is evident that the practical SCF mMIMO systems perform very well since the proposed HA-PMMSE achieves just $ 9\% $ less SE with HA-MMSE receivers when $ \bar{\kappa} =0.06$. When it comes to the HU decoders, the rates worsen with increasing the additive HWIs as well. Regarding the difference between the HA-PMMSE and PMMSE decoders is significant since it is approximately $ 2.6 $ bit/s/Hz for all values of $ \bar{\kappa} $. Also, we illustrate the superior performance of MMSE-style decoders against the MRC decoder as anticipated. Moreover, we provide the MMSE decoder given by Proposition \ref{Proposition:Upperbound}. The provided tightness of the DE analytical expressions with respect to the simulation results, represented by the "cross" symbol, verifies Theorem~\ref{theorem:ULDEMMSE}.
	
	Fig.~\ref{Fig4} compares the SEs versus the number of APs $ M $ in the HA and HU cases for both scalable and conventional CF mMIMO systems. We consider the impact of PN while the additive HWIs are not taken into account. In all cases, $ \mathrm{\overline{SE}}_{k}^{\mathrm{lo}} $ increases with $ M $. Moreover, we observe that SLOs achieve a better SE than the CLO design. Actually, the performance gap between the SLOs and CLO design increases with $ M $ because the phase drifts are independent and in the large system limit, they are averaged, which is a benefit paid as a trade of a higher deployment cost \cite{Bjornson2015,Papazafeiropoulos2017a}.

	Fig.~\ref{Fig5} depicts the impact of the additive HWIs on SCF mMIMO systems with a focus on the proposed HA-PMMSE combiner by showing the average achievable SE per UE versus the number of APs while the PN is assumed negligible. Obviously, the achievable SE increases $ M $ but the additive HWIs at the APs side and UE sides behave differently. Of course, the presence of HWIs degrades the performance. Specifically, the additive HWIs at the transmit side, i.e., at the UE, have a more significant impact on the SE. For example, at $ M=400 $ APs, the loss due to the receive distortion is only $ 4\% $, while in the case of the transmit distortion, the degradation is $ 11\% $. Hence, the quality of hardware at the transmitter side should be considered more during the design.

	\section{Conclusion} \label{Conclusion} 
	In this paper, we investigated the impact of HWIs in SCF mMIMO systems, where the complexity at each AP is finite even when the number of UEs increases to infinity. Specifically, we introduced a general model with HWIs during the channel estimation, and we showed how HWIs modify the CSI. Moreover, we derived the HA-PMMSE combiner considering both HWIs and scalability design issues. Next, we obtained upper and lower bounds on the SE. In particular, we derived the DE of the achievable uplink SE, carrying the impact of HWIs, and we demonstrated the realistic performance of SCF mMIMO systems by varying the fundamental system parameters. This allowed us to extract insightful conclusions as guidelines for the practical implementation of these systems. Among the observations, we would like to highlight that the HA receivers outperform those which are hardware unaware of the presence of both additive HWIs and PN. Also, PMMSE should be preferred due to its advantages and because its performance loss is negligible compared to full MMSE not only in ideal scenarios but in realistic conditions where HWIs exist. Furthermore, we showed that even in SCF mMIMO systems, the HWIs result in ceilings at the SE, and that SCF mMIMO carries its benefits even under practical conditions. Future works should take into account the effect of limited capacity fronthaul links. Another interesting direction could be the extension of this work to full-duplex systems where self-interference arises as another hardware impairment \cite{Zhang2020}. \begin{appendices}
		\section{Proof of Proposition~\ref{Proposition:Upperbound}}\label{Proposition2}
		We first assume that a genie is providing the received with perfect CSI and with the value of the multi-user interference term, which can then be subtracted from the received signal. What remains is a channel where the sum of the transmit and receive distortion, and noise is an independent circularly symmetric complex Gaussian distributed signal. As in~\cite{Telatar1999}, it follows that Gaussian signaling is optimal. An upper bound on the uplink capacity is
		\begin{align}
			R_{k}=\frac{1}{\tau_{\mathrm{p}}}\!\!\sum_{n=1}^{\tau_{\mathrm{c}}-\tau_{\mathrm{p}}}\!\EE\left\{\max_{\bv_{k,n}:\|\bv_{k,n}\|=1}\log_{2}\left (1+\gamma_{k,n}\left (\bv_{k,n}\right )\right )\right\}\!\!,\label{eq46}
		\end{align}
		where we have obtained the upper bound for each $ n $ in the transmission phase, and then, we have taken the average over these bounds~\cite{Pitarokoilis2015}. The uplink SINR in \eqref{eq46} is
		\begin{align}
			\!\!\gamma_{k,n}\!\left( \bv_{k,n} \right) \!=\!\frac{\rho_{k}\bv_{k,n}^{\H}\bD_{k}{\bh}_{k,n}{\bh}_{k,n}^{\H}\bD_{k}\bv_{k,n}}{\bv_{k,n}^{\H}\bD_{k}\tilde{\bSigma}^{\dagger}\bD_{k}\bv_{k,n}},\label{upper2}
		\end{align}
		where $ \tilde{\bSigma} $, $ \bkappa_{\mathrm{r}}^{2}$ and $ \bF_{|\bh_{i}|^{2}}$ are defined in Proposition \ref{Proposition:Upperbound}. We have considered the pseudo-inverse of $ \tilde{\bSigma} $ instead of its inverse because the terms inside the parentheses may not be strictly positive definite \cite{Bjoernson2019a}.

		The maximization of the achievable rate is obtained by noticing that the logarithm is a monotonically increasing function. Indeed, the optimization concerns a generalized Rayleigh quotient problem which is solved by~\eqref{OptimalV} and achieves the upper bound given by~\eqref{Upper1} after plugging the maximizing combiner into~\eqref{eq46}.
		
		\section{Proof of Lemma~\ref{MSElemma}}\label{Lemma3}
		The MSE definition between the estimated received signal, given by~\eqref{estimatedSignal}, and the transmit signal during the training phase $ \left(n \in \{0,\tau\}\right)$ gives
		\begin{align}
			&\mathrm{MSE}_{k,n}=\EE\left\{\|\hat{s}_{k,n}-{s}_{k,n}\|^{2}_{2}\big|\hat{\bH}_{n}\right\}\\
			&=\tr\left(\EE\left\{\left(\hat{s}_{k,n}-{s}_{k,n}\right)\left(\hat{s}_{k,n}-{s}_{k,n}\right)^{\H}\big|\hat{\bH}_{n}\right\}\right)\\
			&=\tr\Bigg({{\bv}_{k,n}^{\H}\bD_{k}}{}\Big(\displaystyle\sum_{i=1}^{K}\rho_{i}\left (1+\kappa_{\mathrm{t}}^{2}\right )\!\!\left(\hat{\bh}_{i,n}\hat{\bh}_{i,n}^{\H}+\tilde{\bR}_{i}\right)\nn\\&+\bkappa_{\mathrm{r}}^{2}\sum_{i=1}^{K}\rho_{i}\left( \bF_{|\hat{\bh}_{i}|^{2}}+\bF_{|{\tilde{\bR}}_{i}|^{2}}\right)+\xi\Id_{W} \Big)\bD_{k}{\bv}_{k,n}\nn\\
			&-\rho_{k}{\bv}_{k,n}^{\H}\bD_{k}\hat{\bh}_{k,n}-\rho_{k}\hat{\bh}_{k,n}^{\H}\bD_{k}{\bv}_{k,n}+\rho_{k}\Big), %\label{errorcovariance}
		\end{align}
		where we have considered that $ \tr \!\left(\!\EE\!\left\{\! {\bh}_{i,n}{\bh}_{i,n}^{\H}\big|\hat{\bh}_{i,n}\!\right\}\!\right)\!=\!\tr\!\left(\!\hat{\bh}_{i,n}\hat{\bh}_{i,n}^{\H}\!+\!\tilde{\bR}_{i}\!\right)$ with $\tilde{\bR}_{i}\!=\! \EE\!\left\{\! \tilde{\bh}_{i,n}\tilde{\bh}_{i,n}^{\H}|\hat{\bh}_{i,n}\!\right\} $.
		%\section{Proof of Proposition~\ref{MMSE_decoder}}\label{Proposition3}
		%The derivation of the HA-PMMSE relies on the minimization of the $ \mathrm{MSE}_{k} $. Taking the derivative of the expression in~\eqref{errorcovariance} with respect to $ {\bv}_{k,n} $ and setting it to zero, we result in the optimal decoder given by~\eqref{OptimalMMSE} after taking into consideration the scalability as described in~\eqref{OptimalV1} .
		
		\section{Proof of Proposition~\ref{LowerBound1}}\label{Proposition4}
		First, we follow the approach in~\cite{Pitarokoilis2015} to compute the average achievable SE for each $ n $ in the transmission phase, and then, we obtain the average over these SEs. The achievable SE per UE in~\eqref{LowerBound} is obtained by taking into account the Gaussianity of the input symbols and by making a worst-case assumption regarding the computation of the mutual information~\cite[Theorem $1$]{Hassibi2003}, where the inter-user interference and the distortion noises are treated as independent Gaussian noise. Moreover, based on \cite{Medard2000}, the receiver treats the channel as deterministic with the gain $ \EE\{\bv_{k,n}^{\H}\bD_{k}\bTheta_{k,n}{\bh}_{k}\} $ in the detection. In contrast, the
		deviation from the average effective
		channel gain is treated as worst-case Gaussian noise. Thus, we obtain
		$ \tilde{\gamma}^{\mathrm{lo}} $, where the expectation operator in the various terms is taken concerning the channel
		vectors as well as the noise processes, which concludes the proof.	
		\section{Proof of Theorem~\ref{theorem:ULDEMMSE}}\label{theorem1}
		First, we obtain the DE of the desired signal power given by~\eqref{sig11}. Specifically, we have
		\begin{align}
			\bv_{k,n}^{\H}\bD_{k}{\bh}_{k,n}&=\hat{\bh}_{k,n}^{\H}\bD_{k}{\bSigma}\bD_{k} {\bh}_{k,n}\label{sig1}\\
			&=\frac{\frac{1}{W}\hat{\bh}^\H_{k,n} \bD_{k}{\bSigma}_k \bD_{k}\left(\hat{\bh}_{k,n}+\tilde{\bh}_{k,n} \right)}{1+\frac{q}{W}{\hatvh}^\H_{k,n}\bD_{k}{\bSigma}_k\bD_{k} {\hatvh}_{k,n} }\label{sig2}\\
			&\asymp\frac{\frac{1}{W}\hat{\bh}^\H_{k,n}\bD_{k} {\bSigma}_k\bD_{k} \hat{\bh}_{k,n} }{1+\frac{q}{W}{\hatvh}^\H_{k,n}\bD_{k}{\bSigma}_k\bD_{k}{\hatvh}_{k,n} }\label{sig3},
		\end{align}
		where in~\eqref{sig1}, we have replaced the expression of the MMSE decoder given by~\eqref{OptimalMMSE2}. 
		%Next, in ~\eqref{sig2}, we have first applied \cite[Thm. 3.7]{Hoydis2013} by taking into account for the independence between $ \hat{\bh}_{k,0} $ and $ \tilde{\bee}_{mk,0} $, and then \cite[Eq.~2.2]{Bai1}. 
		In the next equation, we have applied the matrix inversion lemma and used the orthogonality between the channel and its estimated version. Note that $ q=\left (1+\kappa_{\mathrm{t}_\mathrm{UE}}^{2}\right )\rho_{k} $ while $ {\bSigma}_k^{\dagger} $ is defined as
		\begin{align}
			{\bSigma}_k^{\dagger}\! &=\!{\bSigma}^{\dagger}-\frac{q}{W}\bD_{k}\hat{\bh}_{k,n}\hat{\bh}_{k,n}^{\H}\bD_{k}\\
			&\!=\!\displaystyle\bD_{k}\!\!\left(\sum_{\substack{i \in \mathcal{P}_{k}\nn\\
					i \ne k}}\!\!\frac{q}{W}\hat{\bh}_{i,n}\hat{\bh}_{i,n}^{\H}\!+\!\frac{\bkappa_{\mathrm{r}}^{2}}{W}\sum_{i \in \mathcal{P}_{k}}\!\rho_{i}\bF_{|\hat{\bh}_{i}|^{2}}\!+\!{\al \xi}\Id_{W}\!\!\right)\!\!\bD_{k}\!+\!\tilde{\bDelta}.\nn
		\end{align}
		In the numerator of~\eqref{sig3}, we have applied \cite[Lem. B.26]{Bai2010}. We continue with the use of \cite[Lem. 14.3]{Bai2010} known as rank-1 perturbation lemma, \cite[Lem. B.26]{Bai2010}, and \cite[Theorem 1]{Wagner2012} as\footnote{It is worthwhile to mention that the diagonal matrix $ \bF_{|\hat{\bh}_{i}|^{2}} $ inside the MMSE decoder is considered a deterministic matrix with entries in the diagonal elements the limits of the individual diagonal elements~\cite{Papazafeiropoulos2017a}. In particular, exploiting the uniform convergence $ \lim \sup_W \max_{1\le i\le W} { \left|\left[\hat{\bh}_{i}\hat{\bh}_{i}^{\H} \right]_{ww} - \left[\hat{\bPhi}_{i}\right]_{ww}\right| } = 0$, we have $\left\|\frac{1}{W} \mathrm{diag}(\hat{\bh}_{i} \hat{\bh}_{i}^{H})-\frac{1}{W}\tr\left(\diag\left(\hat{\bPhi}_{i}\right)\right)\right\| \xrightarrow[ W\rightarrow \infty]{\mbox{a.s.}} 0 $.}
		\begin{align}
			\bv_{k,n}^{\H}\bD_{k}\bh_{k,n}
			&\asymp \frac{\frac{1}{W}\tr \left( \bD_{k}\bPhi_{k}\bD_{k}\bT\right)}{1+\frac{q}{W}\tr\left(\bD_{k}\bPhi_{k}\bD_{k}\bT\right)}\\
			&=\frac{\tilde{\delta}_{k}}{1+\delta_{k}},\label{desired2}
		\end{align}
		where $ \tilde{\delta}_{k} =\frac{1}{W}\tr\bD_{k} \bPhi_{k}\bD_{k}\bT$ and $ \delta_{k}=\frac{q}{W}\tr\bD_{k}\bPhi_{k}\bD_{k}\bT $.
		%where in~\eqref{Int1_1}, we have applied the matrix inversion lemma, while, in~\eqref{Int1_2}, we have applied~\cite[Lemma~1]{Truong2013}. The final step in~\eqref{Int1_2} includes applications of \cite[Theorem 2]{Hoydis2013}, and \cite[Lem. 10]{Krishnan2015}.\\
		The term, concerning the deviation from the average effective
		channel gain, becomes
		\begin{align}
			&			\!\!\mathrm{Var}\left\{\bv_{k,n}^{\H}\bD_{k}{\bh}_{k,n}\right\}\asymp\frac{\frac{1}{W}\EE \left\{|\hat{\bh}^\H_{k,n}\bD_{k} {\bSigma}_k\bD_{k} {\bh}_{k,n}|^{2}\right\}}{\left(1+\delta_{k}\right)^{2}}\label{Int8_1}\\
			&~~~\asymp\frac{\EE \left\{\frac{1}{W^{2}}\tr \left(\bD_{k}\bPhi_{k} \bD_{k}{\bSigma}_k \bD_{k}\left( \bR_{k} - \bPhi_{k}\right)\bD_{k}\bSigma_{k}\right)\right\}}{\left(1+\delta_{k}\right)^{2}}\label{Int8_3}\\
			&~~~\asymp\frac{\EE \left\{\frac{1}{W^{2}}\tr \left(\bD_{k}\left( \bR_{k} - \bPhi_{k}\right)\bD_{k}\bT^{'}_{k}\right)\right\}}{\left(1+\delta_{k}\right)^{2}},\label{Int8_4}
		\end{align}
		where in \eqref{Int8_1}, we have used the matrix inversion lemma, \cite[Lem. B.26]{Bai2010}, and \cite[Theorem 1]{Wagner2012}. In \eqref{Int8_3}, we have applied the rank-1 perturbation lemma, \cite[Lem. B.26]{Bai2010}, and \cite[Lem. 10]{Krishnan2015}. and \cite[Lem. B.26]{Bai2010} again. The last step includes application of \cite[Theorem 2]{Hoydis2013}.\\
		Based on~\eqref{int1}, the interference power of the $ k $th UE is written as
		\begin{align}
			&|\bv_{k,n}^{\H}\bD_{k}{\bh}_{i,n}|^{2}\asymp \left|\frac{\frac{1}{W}\hat{\bh}^\H_{k,n}\bD_{k}{\bSigma}_k\bD_{k} {\bh}_{i,n} }{1+\delta_{k} }\right|^{2}\label{Int2_1}\\
			&\asymp \frac{\hat{\bh}^\H_{k,n} \bD_{k}{\bSigma}_k \bD_{k}{\bh}_{i,n} \bh_{i,n}^{\H}\bD_{k} {\bSigma}_k\bD_{k}\hat{\bh}_{k,n} }{W^{2}\left(1+\delta_{k}\right)^{2}}\label{Int2_2}\\
			&\asymp \frac{{\bh}^\H_{i,n}\bD_{k} {\bSigma}_k\bD_{k}\bPhi_{k}\bD_{k}{\bSigma}_k \bD_{k}{\bh}_{i,n} }{W^{2}\left(1+\delta_{k}\right)^{2}}\label{Int2_3}\\
			&\!\asymp\!\frac{ 1}{\left(1\!+\!\delta_{k}\right)^{2}}\bigg(\!\!\frac{1}{W^{2}}{\bh}^\H_{i,n}\bD_{k} {\bSigma}_{ki}\bD_{k}\bPhi_{k}\bD_{k}{\bSigma}_{ki}\bD_{k} {\bh}_{i,n}\!\nn\\&+\!\frac{|{\bh}^\H_{i,n} \bD_{k}{\bSigma}_{ki}\bD_{k}\hat{\bh}_{i,n}|^{2}\hat{\bh}^\H_{i,n}\bD_{k} {\bSigma}_{ki}\bD_{k}\bPhi_{k}\bD_{k}{\bSigma}_{ki}\bD_{k} \hat{\bh}_{i,n}}{W^{4}\left(1+\delta_{i}\right)^{2}}\nn\\
			&-2\mathrm{Re}\bigg\{
			\frac{\left(\hat{\bh}^\H_{i,n}\bD_{k}{\bSigma}_{ki}\bD_{k}{\bh}_{i,n}\right)}{W^{3}\left(1+\delta_{i}\right)^{2}} \left({\bh}^\H_{i,n} \bD_{k}{\bSigma}_{ki}\bD_{k}\bPhi_{k}\bD_{k}{\bSigma}_{ki}\bD_{k} \hat{\bh}_{i,n}\right)\!\!\!\bigg\}\!\bigg)\label{Int2_4}\\
			&\asymp \frac{1}{\left(1+\delta_{k}\right)^{2}}\bigg(\zeta_{ki}+\frac{|\nu_{ki}|^{2}\mu_{ki}}{\left(1+\delta_{i}\right)^{2}}-2\mathrm{Re}\left\{
			\frac{\nu_{ki}^{*}\mu_{ki}}{\left(1+\delta_{i}\right)}\right\}\bigg),\label{Int2_5}
		\end{align}
		where we have applied the matrix inversion lemma in~\eqref{Int2_1}, while in~\eqref{Int2_2} and \eqref{Int2_3}, we have used \cite[Lem. B.26]{Bai2010}. In \eqref{Int2_4}, we have applied again the matrix inversion lemma, and in the last step, we have used the rank-1 perturbation lemma, \cite[Lem. B.26]{Bai2010}, \cite[Theorem 1]{Wagner2012}, and \cite[Theorem 2]{Hoydis2013}. The definitions of the various parameters are given in the presentation of the theorem. 
		Below, the terms, corresponding to the transmit and
		receive distortions, are derived. Specifically, we have
		\begin{align}
			&\sigma_{\mathrm{t}}^{2}= \sum_{i=1}^{K}\kappa_{\mathrm{t}}^{2}\rho_{i}\EE\left\{|\bv_{k,n}^{\H}\bD_{k}\bh_{i,n}|^{2} \right\}\nn\\
			&=\kappa_{\mathrm{t}}^{2}\left(\rho_{k}\EE\left\{|\bv_{k,n}^{\H}\bD_{k}\bh_{k,n}|^{2}\right\}+\sum_{i\ne k}^{K}\rho_{i}\EE\left\{|\bv_{k,n}^{\H}\bD_{k}\bh_{i,n}|^{2}\right\}\right)\nn\\
			& \asymp\! \frac{\kappa_{\mathrm{t}}^{2}}{\left(1\!+\!\delta_{k}\right)^{2}}\!\!\left(\!\rho_{k}{\tilde{\delta}^{2}_{k}}{}\!+\!\sum_{i\ne k}^{K}\!\rho_{i}\!\!\left(\!\zeta_{ki}\!+\!\frac{|\nu_{ki}|^{2}\mu_{ki}}{\left(1\!+\!\delta_{i}\right)^{2}}\!-\!2\mathrm{Re}\!\left\{\!
			\frac{\nu_{ki}^{*}\mu_{ki}}{\left(1\!+\!\delta_{i}\right)}\!\right\}\!\!\right)\!\!\right)\!\!,\label{Int6}
		\end{align}
		where, in \eqref{Int6}, we have substituted \eqref{desired2} and \eqref{Int2_5}.
		Also, the deterministic $ \sigma_{\mathrm{r}}^{2} $ as $ W \to \infty $ becomes
		\begin{align}
			\!\!\!\sigma_{\mathrm{r}}^{2}\!&=\!
			\EE\left\{\bv_{k,n}^{\H}\bD_{k}\bkappa_{\mathrm{r}}^{2}\!\left(\Id_{W}\circ\bH_{n}\bP \bH^{\H}_{n}\right)\bD_{k} \bv_{k,n}\right\}\label{sigmar}\\
			&\!\!\!\asymp \!\frac{\bkappa_{\mathrm{r}}^{2}\EE\left\{\tr\left(\bD_{k} \left(\Id_{W}\circ \bH_{n} \bP\bH^{\H}_{n} \right)\bD_{k}\bSigma_{k}\bD_{k}\bPhi_{k}\bD_{k}\bSigma_{k}\right)\right\}}{W^{2}\left(1+\delta_{k}\right)^{2}}\nn\\
			&\!\!\!\asymp\! \frac{\eta_{k}^{'}}{W^{2}\left(1+\delta_{k}\right)^{2}}\sum_{i=1}^{K}\rho_{i}\tr\left( \bD_{k}^{2}\bkappa_{\mathrm{r}}^{2}\bR_{i}\right)\!,\label{Int7_2}
		\end{align}
		where in~\eqref{sigmar} we have written the diagonal matrix in terms of a Hadamard product. Next, we have exploited the freeness between $\bv_{k,n}\bv_{k,n}^{\H} $ and the diagonal matrix $\Id_{W}\circ\bH \bH^{\H} $. Also, we have applied \cite[Lem. B.26]{Bai2010}\cite[Theorem 1]{Wagner2012}, \cite[Theorem 2]{Hoydis2013}, and~\cite[p. 207]{Tao2012} while we have set $ \eta^{'}_{k}=\frac{1}{W}\tr\bT^{'}_{k}$. 
		The DE of the last term of~\eqref{int1}, corresponding to the ATN contribution, becomes
		\begin{align}
			\|\bv_{k,n}^{\H}\bD_{k}\|^{2}&\asymp\frac{\frac{1}{W^{2}}\hat{\bh}_{k,n}^{\H}\bD_{k}\bSigma_{k}\bD_{k}^{2}\bSigma_{k}\bD_{k}\hat{\bh}_{k,n}}{\left(1+\delta_{k}\right)^{2}}\label{Int4_1}\\
			&\asymp\frac{\frac{1}{W^{2}}\tr\left(\bSigma_{k}\bD_{k}^{2}\bSigma_{k}\bD_{k}\bPhi_{k}\bD_{k}\right)}{\left(1+\delta_{k}\right)^{2}}\label{Int4_2}\\
			&=\frac{\frac{1}{W^{2}}\tr\left(\bT^{''}_{k}\right)}{\left(1+\delta_{k}\right)^{2}}\label{Int4_3},
		\end{align}
		where in~\eqref{Int4_1} and~\eqref{Int4_2}, we have applied the matrix inversion lemma and \cite[Lem. B.26]{Bai2010}, respectively. In the last step, we have used the rank-1 perturbation lemma as well as \cite[Theorem 2]{Hoydis2013} and \cite[Theorem 1]{Wagner2012}. 
		By substituting \eqref{desired2}, \eqref{Int8_4}, \eqref{Int2_5}, \eqref{Int6}, \eqref{Int7_2}, and \eqref{Int4_3} into \eqref{sig11} and \eqref{int1}, the DE SINR is derived and the proof is concluded.

	\end{appendices}
	\bibliographystyle{IEEEtran}
	
	\bibliography{mybib}
	\bibliography{mybib}
\end{document}